\documentclass[lettersize,journal]{IEEEtran}
\usepackage{amsmath,amssymb,amsfonts,mathrsfs,commath,mathtools}
\usepackage{algorithmic}
\usepackage{booktabs}
\usepackage{multirow}
\usepackage{algorithm}
\usepackage{array}
\usepackage[caption=false,font=normalsize,labelfont=sf,textfont=sf]{subfig}
\usepackage{textcomp}
\usepackage{stfloats}
\usepackage{url}
\usepackage{verbatim}
\usepackage{graphicx}
\usepackage{cite}

\usepackage{graphicx}
\usepackage{epstopdf}
\usepackage{amsthm}

\allowdisplaybreaks

\DeclareMathOperator{\sgn}{sgn}

\DeclareMathOperator{\esup}{ess\, sup}



\providecommand*{\ped}[1]{%
\ensuremath{_\textnormal{#1}}}

\providecommand*{\eu}%
{\ensuremath{\mathrm{e}}}
\providecommand*{\im}%
{\ensuremath{\mathrm{j}}}
\providecommand*{\GammaF}%
{\ensuremath{\mathrm{\Gamma}}}
\providecommand*{\BetaF}%
{\ensuremath{\mathrm{\Beta}}}

\newtheorem{remark}{Remark}[section]
\newtheorem{lemma}{Lemma}[section]
\newtheorem{theorem}{Theorem}[section]

\newtheorem{definition}{Definition}[section]
\newtheorem{corollary}{Corollary}[section]
\usepackage[table]{xcolor} 
\usepackage{hyperref}
\hypersetup{
    colorlinks = true,
    linkbordercolor = {white},
            linkcolor=blue,
            bookmarks=true,
            filecolor=blue,
            urlcolor=blue,
            citecolor=blue
}

\hypersetup{hidelinks=false}
\usepackage{textcomp}

\begin{document}

\title{First-order friction models with bristle dynamics:\\ lumped and distributed formulations}

\author{Luigi Romano, \IEEEmembership{Member, IEEE}, Ole Morten Aamo, \IEEEmembership{Senior Member, IEEE}, Jan Åslund, and Erik Frisk 
\thanks{This research was financially supported by the project FASTEST (Reg. no. 2023-06511), funded by the Swedish Research Council.}
\thanks{L. Romano is with the Department of Electrical Engineering, Linköping University, Linköping, Sweden, and the Department of Engineering Cybernetics, NTNU, Trondheim, Norway (e-mail: luigi.romano@liu.se).}
\thanks{O. M. Aamo is with the Department of Engineering Cybernetics, NTNU, Trondheim, Norway (e-mail: ole.morten.aamo@ntnu.no).}
\thanks{J. Åslund and E. Frisk are with the Department of Electrical Engineering, Linköping University, Linköping, Sweden (e-mail: jan.aslund@liu.se and erik.frisk@liu.se).}
\thanks{Manuscript received April 19, 2021; revised August 16, 2021.}}

\markboth{Journal of \LaTeX\ Class Files,~Vol.~14, No.~8, August~2021}%
{Luigi Romano \MakeLowercase{\textit{et al.}}: First-order friction models with bristle dynamics: lumped and distributed formulations}


\maketitle

\begin{abstract}
Dynamic models, particularly rate-dependent models, have proven effective in capturing the key phenomenological features of frictional processes, whilst also possessing important mathematical properties that facilitate the design of control and estimation algorithms. However, many rate-dependent formulations are built on empirical considerations, whereas physical derivations may offer greater interpretability.
In this context, starting from fundamental physical principles, this paper introduces a novel class of first-order dynamic friction models that approximate the dynamics of a bristle element by inverting the friction characteristic. Amongst the developed models, a specific formulation closely resembling the LuGre model is derived using a simple rheological equation for the bristle element. This model is rigorously analyzed in terms of stability and passivity -- important properties that support the synthesis of observers and controllers. Furthermore, a distributed version, formulated as a \emph{hyperbolic partial differential equation} (PDE), is presented, which enables the modeling of frictional processes commonly encountered in rolling contact phenomena. The tribological behavior of the proposed description is evaluated through classical experiments and validated against the response predicted by the LuGre model, revealing both notable similarities and key differences.
\end{abstract}

\begin{IEEEkeywords}
Friction, stability, passivity, distributed parameter systems, hyperbolic PDEs
\end{IEEEkeywords}

\section{Introduction}
Friction is a fundamental physical phenomenon that profoundly influences the behavior of a broad spectrum of mechanical systems, ranging from servo mechanisms to pneumatic and hydraulic actuators \cite{Flores1,Flores2,Flores3,NonlinearDynId,2D}. Whilst it plays a crucial role in enabling force transmission and traction generation, friction also introduces undesirable effects in control systems -- such as tracking errors, stick-slip oscillations, and limit cycles -- that can significantly degrade performance, especially in high-precision or low-velocity applications \cite{Motors,Hydraulics,Hydraulics2,Compensation}. As a result, friction estimation and compensation remain persistent challenges \cite{CST1,CST2,CST3,CST4}, necessitating accurate and reliable models that capture its complex phenomenology. Traditional static representations, such as Coulomb, viscous, and Stribeck friction models, often fall short in describing the nuanced and dynamic behavior of friction in these demanding scenarios \cite{2D,Antali}.

In response to these limitations, the past few decades have witnessed a progressive shift toward dynamic friction models capable of capturing key phenomena such as pre-sliding hysteresis, velocity-dependent memory, and time-lag effects. This evolution has spanned disciplines as diverse as geophysics and electromechanics, gaining renewed momentum within the systems and control community following Armstrong's seminal 1991 review \cite{Armstrong}.

Amongst the various models proposed, the LuGre description \cite{Astrom1,Olsson,Astrom2} represents a pivotal milestone. Drawing inspiration from the rate-dependent formulations developed in geophysics by Rice and Ruina \cite{Rice}, the LuGre model provides a computationally tractable representation that quickly gained popularity due to its ability to qualitatively replicate a wide range of frictional behaviors.
However, despite its merits, the LuGre model fails to capture nonlocal memory effects in the pre-sliding regime – a key limitation when modeling friction near velocity reversals, where the force-displacement relationship is governed primarily by adhesive interactions at asperity contacts. Indeed, in this regime, experimental observations suggest that friction behaves more like a hysteresis function of displacement, rather than velocity. To address this shortcoming, several extensions of the LuGre model have been proposed \cite{Integrated,Leuven,Elasto1,Elasto2,GMS}. For instance, the Leuven model \cite{Integrated,Leuven} incorporates nonlocal hysteresis effects, albeit at the cost of significantly increased implementation complexity. Another line of development led to the so-called generalized Maxwell-slip (GMS) models \cite{GMS}, which embed rate-state dynamics into the slip phases of classical Maxwell-slip blocks. These models offer an improved representation of displacement-dependent hysteresis and nonlocal memory, whilst retaining a relatively simple structure.

More generally, rate-dependent models are typically described by two fundamental equations: an algebraic one, which postulates a suitable rheological model for generating the friction force, and a dynamic one, which governs the time evolution of the frictional state(s). These equations are often formulated independently, to replicate common behaviors observed in mechanical systems, such as hysteresis, frictional lag, and stick-slip \cite{Fal,Fal2}. However, physically grounded justifications for these models are essential, as they provide greater interpretability and clarity, as emphasized in \cite{GMS,Fal3}. In this context, recent studies \cite{Rill} have demonstrated that the LuGre model can be derived by inverting the implicit relationship between displacement and friction force in a simple rheological bristle model. This concept of friction inversion is not entirely novel and has proven to be instrumental in modeling transient rolling contact phenomena, which typically necessitate distributed parameter formulations \cite{LibroMio,Two-regime}.

Building on this foundation, and inspired by the ideas contained in \cite{Rill}, this paper introduces a novel class of rate-dependent dynamical friction models derived by systematic inversion of the friction characteristic. This approach yields a family of physically motivated models that are both interpretable and amenable to rigorous mathematical analysis. Whereas the proposed approach may be, in principle, applied to any existing model, a specific variant, based on a linear spring-damper constitutive law for the bristle element, is derived that retains the structural simplicity of the LuGre model. The model's tractability facilitates stability, dissipativity, and passivity analysis, and it naturally extends to a distributed formulation suitable for the description of rolling contact phenomena. As for its lumped counterpart, the behavior of this distributed variant is rigorously analyzed in this paper in terms of stability and passivity. 

More explicitly, the main contributions of this work are:
\begin{enumerate}
\item The introduction of a general procedure to derive physically-consistent lumped and distributed first-order dynamic friction models starting from rheological descriptions of bristle elements, 
\item The development and analysis of a specific friction model inspired by the LuGre formulation, in both its lumped and distributed variants,
\item The experimental validation of the proposed model in its lumped form.
\end{enumerate}

The choice of using the LuGre as a benchmark is simple: despite its limitations, it remains the most widely adopted rate-dependent formulation encountered in the control literature, as well as in the modeling of many mechanical systems \cite{2D,Antali}. In particular, its enduring popularity can be attributed to its mathematically elegant structure and its capacity to serve as a versatile heuristic framework. In the same context, it is also worth mentioning that distributed parameter extensions to the LuGre model have been successfully introduced to describe rolling contact processes in \cite{TsiotrasConf,Tsiotras1, Tsiotras2,Deur1,Deur2,LuGreSpin,CarcassDyn}, providing a fertile terrain for comparison with the distributed formulation proposed in this paper. The stability and passivity analyses conducted in this work follow instead the approach of \cite{LuGreDistr}, but with important differences that highlight both the distinct dynamical behavior of the proposed model and the implications for its parametrization. In particular, it is revealed that passivity holds virtually for every combination of model parameters, as opposed to the LuGre model, which requires postulating the damping coefficient as velocity-dependent.

The remainder of this manuscript is organized as follows. Section~\ref{Sect:Gen} first reviews the general mathematical structure of rate-dependent models, and then discusses how to derive their governing equations from a physical approximation of the friction characteristic. A lumped, LuGre-inspired formulation, built according to the procedure outlined in Section~\ref{Sect:Gen}, is then introduced in Section~\ref{eq:LumpedLuGre}; the corresponding distributed extension is deduced in Section~\ref{sect:Distr}. The tribological response of the proposed lumped model is then compared against that of the original LuGre formulation in Section~\ref{sect:exper}, where its experimental validation is also addressed. Finally, the main conclusions of the paper are summarized in Section~\ref{ref:Concl}, where future directions for research are also indicated.

\subsection*{Notation and preliminaries}
In this paper, $\mathbb{R}$ denotes the set of real numbers; $\mathbb{R}_{>0}$ and $\mathbb{R}_{\geq 0}$ indicate the set of positive real numbers excluding and including zero, respectively. The set of positive integer numbers is indicated with $\mathbb{N}$, whereas $\mathbb{N}_{0}$ denotes the extended set of positive integers including zero, i.e., $\mathbb{N}_{0} = \mathbb{N} \cup \{0\}$.
$L^2((0,1);\mathbb{R})$ denotes the Hilbert space of square-integrable functions on $(0,1)$ with values in $\mathbb{R}$, endowed with inner product $\langle u, v \rangle_{L^2((0,1);\mathbb{R})} = \int_0^1 u(\xi)v(\xi) \dif \xi$ and induced norm $\norm{u(\cdot)}_{L^2((0,1);\mathbb{R})}^2 \triangleq \langle u, u \rangle_{L^2((0,1);\mathbb{R})}$. $H^1((0,1);\mathbb{R})$ is the Sobolev space of functions $u \in L^2((0,1);\mathbb{R})$ whose weak derivative also belongs to $L^2((0,1);\mathbb{R})$; it is a Hilbert space equipped with norm $\norm{u(\cdot)}_{H^1((0,1);\mathbb{R})}^2 \triangleq \norm{u(\cdot)}_{L^2((0,1);\mathbb{R})}^2+ \norm{\pd{u(\cdot)}{\xi}}_{L^2((0,1);\mathbb{R})}^2$. $C^k(\overline{\Omega};\mathcal{X})$, $k\in \{0,1,\dots,\infty\}$, denotes the spaces of $k$-times continuously differentiable functions on $\overline{\Omega}$ with values in $\mathcal{X}$, where $\overline{\Omega}$ is the closure of a bounded domain $\Omega$, and $\mathcal{X}$ may be a Banach space or a structured subset of $\mathbb{R}$ (for $T = \infty$, the closure of $(0,T)$ is identified with $\mathbb{R}_{\geq 0}$). For a function $f :\Omega \mapsto \mathbb{R}$, the sup norm is defined as $\norm{f(\cdot)}_\infty \triangleq \esup_{\Omega} \abs{f(\cdot)}$; $f : \Omega \mapsto \mathbb{R}$ belongs to the space $L^\infty(\Omega;\mathbb{R})$ if $\norm{f(\cdot)}_\infty < \infty$. A function $f \in C^0(\mathbb{R}_{\geq 0}; \mathbb{R}_{\geq 0})$ belongs to the space $\mathcal{K}$ if it is strictly increasing and $f(0) = 0$; $f \in \mathcal{K}$ belongs to the space $\mathcal{K}_\infty$ if it is unbounded. Finally, a function $f \in C^0(\mathbb{R}_{\geq 0}^2; \mathbb{R}_{\geq 0})$ belongs to the space $\mathcal{KL}$ if $f(\cdot,t) \in \mathcal{K}$ and is strictly decreasing in its second argument, with $\lim_{t\to \infty}f(\cdot,t) = 0$. 

An application of Theorem~\ref{thm:Theorem1} below will guide the derivation of the rate-dependent friction models developed in the paper.

\begin{theorem}[Edwards \cite{Edwards}]\label{thm:Theorem1}
Suppose that the mapping $H : \mathbb{R}^{m+n}\mapsto \mathbb{R}^n$ is $C^1$ in a neighborhood of a point $(x^\star,y^\star)$, where $H(x^\star,y^\star) = 0$. If the Jacobian matrix $\nabla_{y}H(x^\star,y^\star)^{\mathrm{T}}$ is nonsingular, there exist a neighborhood $\mathcal{X}$ of $x^\star$ in $\mathbb{R}^m$, a neighborhood $\mathcal{Y}$ of $(x^\star,y^\star)$ in $\mathbb{R}^{m+n}$, and a mapping $h \in C^1(\mathcal{U};\mathbb{R}^n)$ such that $y =h(x)$ solves the equation $H(y,x) = 0$ in $\mathcal{Y}$. 
In particular, the implicitly defined mapping $h(\cdot)$ is the limit of the sequence $\{h_k\}_{ k \in \mathbb{N}_0}$ of the successive approximations inductively defined by
\begin{subequations}
\begin{align}
h_{k+1}(x) & = h_k(x) - \nabla_{y}H(x^\star,y^\star)^{-\mathrm{T}}H\bigl(x,h_k(x)\bigr), \\
h_0(x) & = y^\star,
\end{align}
\end{subequations}
for $x\in \mathcal{X}$.
\end{theorem}

\section{Basic considerations on rate-dependent friction models}\label{Sect:Gen}
The present section introduces a general methodology to derive the governing equations of rate-dependent models via the inversion of the friction characteristic. To this end, the general structure of existing rate-dependent models is first discussed in Section~\ref{ref:GenStruc}. The proposed approach is then illustrated in very general terms in Section~\ref{sect:impl}.

\subsection{General structure of rate-dependent models}\label{ref:GenStruc}

Existing first-order rate-dependent friction models available from the literature have the following form:
\begin{subequations}\label{eq:RateDep}
\begin{align}
F\ped{b} & = f(\dot{z},z,v), \label{eq:FbRate} \\
\dot{z} & = h(z,v), &&t \in (0,T), \label{eq:FbRateZ}
\end{align}
\end{subequations}
where $z \in \mathbb{R}$ represents the frictional state, $v \in \mathbb{R}$ denotes the relative velocity between the contacting bodies, and $F\ped{b} \in \mathbb{R}$ is the friction force. Equation~\eqref{eq:FbRate} is a (rate-dependent) algebraic relationship postulating the friction force as a function of the variable $z$, its derivative $\dot{z}$, and the relative velocity input $v$;~\eqref{eq:FbRateZ} is an \emph{ordinary differential equation} (ODE) governing the dynamics of the frictional variable $z$. The functions $f:\mathbb{R}^3 \mapsto \mathbb{R}$ and $h : \mathbb{R}^2 \mapsto \mathbb{R}$ are often nonlinear, with $h(\cdot,\cdot)$ possibly discontinuous to capture pre-sliding and gross sliding behaviors. The formulation~\eqref{eq:RateDep} may also be generalized to the case where the scalar frictional variable $z$ is replaced by a vector of variables, to account, for example, for nonlocal hysteretical effects; this is however beyond the scope of the present paper. In the simplest scalar case, $z$ is usually interpreted as the deflection of a bristle element attached to one of the two contacting bodies in relative motion, as discussed more explicitly in Section~\ref{sect:impl}.

Rate-dependent heuristic friction models, as expressed in the form~\eqref{eq:RateDep}, are typically obtained by postulating empirical expressions for the functions $f(\cdot,\cdot,\cdot)$ and $h(\cdot,\cdot)$ to faithfully replicate the phenomenological behaviors observed in mechanical and mechatronic systems. In some cases, the structure of these models is further supported by drawing analogies with physics-based formulations, as it happens, for instance, for the GMS models. In fact, as elucidated in Section~\ref{sect:impl}, first-order rate-dependent models can be directly derived as simplified approximations to more detailed physical descriptions, where rheological equations for the bristle element are combined with suitably designed friction curves.

\subsection{Physical derivation of rate-dependent models}\label{sect:impl}
The governing equations of rate-dependent friction models may be derived by mixing the right doses of physical intuition and algebraic dexterity, providing a formal justification for many of the first-order dynamic equations encountered in the literature. To this end, it is profitable to consider the situation illustrated schematically in Figure~\ref{fig:LumpModel}: a rigid body travels with relative velocity $v$ with respect to a rigid substrate. To the lower boundary of the upper body, deformable bristles are attached, whose deflection is denoted by $z$. The total sliding velocity between the tip of the bristle and the lower body reads
\begin{align}
v\ped{s}(\dot{z},v) = v- \dot{z}. 
\end{align}

\begin{figure}
\centering
\includegraphics[width=0.7\linewidth]{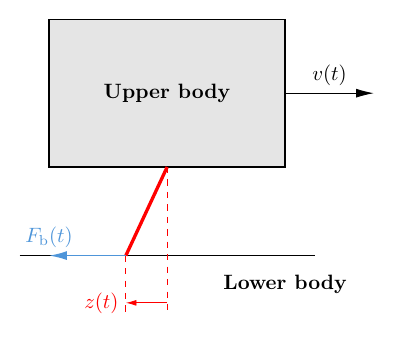} 
\caption{A schematic representation of the friction model.}
\label{fig:LumpModel}
\end{figure}

Neglecting inertial phenomena, the bristle force $F\ped{b} = f(\dot{z},z,v)$, given by~\eqref{eq:FbRate}, must oppose the friction force exerted on the tip of the bristle, which may be generally modeled as
\begin{align}\label{eq:ffric}
F\ped{r}\bigl(v\ped{s}(\dot{z},v)\bigr) = \dfrac{v\ped{s}(\dot{z},v)}{\abs{v\ped{s}(\dot{z},v)}_\varepsilon}\mu\bigl(v\ped{s}(\dot{z},v)\bigr)p,
\end{align}
where $\mu : \mathbb{R} \mapsto [\mu\ped{min},\infty)$, with $\mu\ped{min}\in \mathbb{R}_{>0}$, denotes the friction coefficient, possibly accounting for the Stribeck effect and viscous friction, and $p \in \mathbb{R}_{>0}$ is the normal force acting on the bristle. For instance, adopting a generalized Coulomb friction model, an expression for $\mu(v\ped{s})$ could be
\begin{align}\label{eq:fViscOOO}
\mu(v\ped{s}) = \mu\ped{d} + (\mu\ped{s}-\mu\ped{d})\eu^{-(\abs{v\ped{s}}/v\ped{S})^\delta} + \mu\ped{v}(v\ped{s}),
\end{align}
where $\mu\ped{d}, \mu\ped{s} \in \mathbb{R}_{>0}$ denote the \emph{dynamic} and \emph{static friction coefficient}, respectively, $v\ped{S}\in \mathbb{R}_{\geq 0}$ is the \emph{Stribeck velocity}, $\delta \in \mathbb{R}_{\geq 0}$ the \emph{Stribeck exponent}, and $\mu\ped{v} \in C^1(\mathbb{R};\mathbb{R}_{\geq 0 })$ is the \emph{viscous friction} coefficient. Finally, $\abs{\cdot}_\varepsilon \in C^0(\mathbb{R};\mathbb{R}_{\geq 0})$, with $\varepsilon \in \mathbb{R}_{\geq0}$, is a regularization of the absolute value $\abs{\cdot}$ for $\varepsilon\in \mathbb{R}_{>0}$, often converging uniformly to $\abs{\cdot}$ in $C^0(\mathbb{R};\mathbb{R}_{\geq 0})$ for $\varepsilon \to 0$ (e.g., $\abs{v}_\varepsilon= \sqrt{v^2 +\varepsilon}$), and with $\abs{\cdot}_\varepsilon \in C^1(\mathbb{R};\mathbb{R}_{\geq 0})$ for $\varepsilon \in \mathbb{R}_{>0}$. In the following, the regularization with $\varepsilon \in \mathbb{R}_{>0}$ is introduced to streamline the analysis and avoid excessively technical arguments, but also because the adoption of regularized friction models is extensive in mechanical engineering practice \cite{Rill,FloresReg}.

Two possible non-smooth trends for the friction force $F\ped{r}(v\ped{s})$ are illustrated in Figure~\ref{fig:muCurv}. Generally, the friction coefficient $\mu(\cdot)$ may be postulated to be a discontinuous function, whereas this paper assumes for simplicity $\mu\in C^1(\mathbb{R};[\mu\ped{min},\infty))$.
\begin{figure}
\centering
\includegraphics[width=1\linewidth]{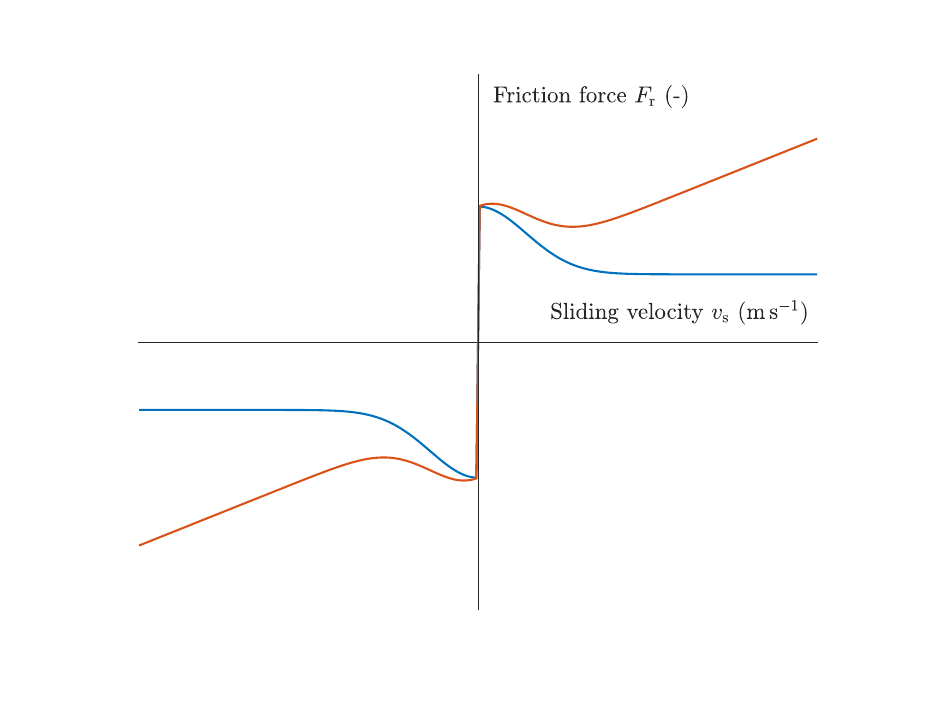} 
\caption{Two possible friction curves: generalized Coulomb (blue line), and generalized Coulomb with viscous friction (orange line).}
\label{fig:muCurv}
\end{figure}

In any case, equating~\eqref{eq:FbRate} and~\eqref{eq:ffric} yields
\begin{align}\label{eq:nonlinImplGen}
\begin{split}
H(\dot{z}, z,v) = f(\dot{z},z,v)- & \dfrac{F\ped{r}\bigl(v\ped{s}(\dot{z},v)\bigr)}{p} = 0, \quad t \in (0,T),
\end{split}
\end{align}
which is an implicit nonlinear ODE for the bristle dynamics $z(t)$. In the sliding regime, where $\abs{\dot{z}} \ll \abs{v}$, and for sufficiently smooth $H(\cdot,\cdot,\cdot)$,~\eqref{eq:nonlinImplGen} may be approximated by invoking Theorem~\ref{thm:Theorem1} with $x \triangleq (z, v)$ and $y\triangleq \dot{z}$, which gives
\begin{align}\label{eq:zDotderGen}
\dot{z}_{k+1} = \dot{z}_k - {\dpd{H(\dot{z}^\star,z^\star, v^\star)}{\dot{z}}}^{-1}H(\dot{z}_k, z, v), \quad k \in \mathbb{N}_0. 
\end{align}
Since $H(\cdot,\cdot,\cdot)$ does not need to be smooth, the following approximation is however adopted: 
\begin{align}\label{eq:approxDfGen}
\begin{split}
\dpd{H(\dot{z}, z, v)}{\dot{z}} &  =\dpd{f(\dot{z},z,v)}{\dot{z}} -\dfrac{1}{p}\dod{F\ped{r}\bigl(v\ped{s}(\dot{z},v)\bigr)}{v\ped{s}}\dpd{v\ped{s}(\dot{z},v)}{\dot{z}} \\
& \approx  \dpd{f(\dot{z},z,v)}{\dot{z}}+ \dfrac{1}{p}\dfrac{F\ped{r}(\dot{z},v)}{v\ped{s}(\dot{z},v)}   \\
& =\dpd{f(\dot{z},z,v)}{\dot{z}}+ \dfrac{\mu\bigl(v\ped{s}(\dot{z},v)\bigr)}{\abs{v\ped{s}(\dot{z},v)}_\varepsilon}.
\end{split}
\end{align}
The approximation committed in~\eqref{eq:approxDfGen} is informally legitimated by the fact that the friction force $F\ped{r}(v\ped{s})$ may exhibit a sharp discontinuity for $v\ped{s} = 0$ \cite{Rill}.
Supposing that there exists a unique $z^\star = z^\star(v)$ solving~\eqref{eq:nonlinImplGen} for $\dot{z}^\star = 0$ and for all $v \in \mathbb{R}$, truncating at $k=1$, and combining~\eqref{eq:zDotderGen} and~\eqref{eq:approxDfGen} with $(\dot{z}^\star, z^\star ,v^\star) = (0,z^\star(v),v)$ provides, for an initial guess $\dot{z}_0 = 0$,
\begin{subequations}\label{eq:ODEzGen}
\begin{align}
& \dot{z}(t) = h\bigl(z(t),v(t)\bigr), \quad  t\in(0,T),\label{eq:ODEz0Gen} \\
& z(0) = z_0,
\end{align}
\end{subequations}
where
\begin{align}
h(z,v) = -\dfrac{\abs{v}_\varepsilon}{g(v)}\bigl( f(0,z,v)-\mu(v)\sgn_\varepsilon(v)\bigr),
\end{align}
with
\begin{align}
\sgn_\varepsilon(v) \triangleq \dfrac{v}{\abs{v}_\varepsilon},
\end{align}
and
\begin{align}\label{eq:g(v)Gen}
g(v) \triangleq \dpd{f\bigl(0,z^\star(v),v\bigr)}{\dot{z}}\abs{v}_\varepsilon + \mu(v).
\end{align}
In the remainder of the paper, the family of rate-dependent models~\eqref{eq:RateDep} with bristle dynamics described by~\eqref{eq:ODEzGen} is referred to as \emph{Friction with Bristle Dynamics} (FrBD) models. Infinitely many variants of the FrBD model may be obtained by prescribing the analytical expression for the force $F\ped{b} = f(\dot{z},z,v)$. Building on a simple rheological model of the bristle element, the next Section~\ref{eq:LumpedLuGre} introduces a specific FrBD formulation inspired by the LuGre model.

Before moving to Section~\ref{eq:LumpedLuGre}, two important considerations are collected below.
\begin{remark}\label{remark:OOOOO}
The procedure outlined above provides a clear methodology to construct the functions $f(\cdot,\cdot,\cdot)$ and $h(\cdot,\cdot)$ appearing in~\eqref{eq:RateDep}: the former may be obtained from a constitutive relationship for the bristle element, whereas the latter requires the inversion of the friction curve according to~\eqref{eq:nonlinImplGen} and~\eqref{eq:zDotderGen}. 
\end{remark}

\begin{remark}\label{remark:OOOOO}
The approximated inversion~\eqref{eq:zDotderGen} is valid under the condition $\abs{\dot{z}} \ll \abs{v}$, which rules out the sticking regime for the bristle dynamics. This is in theoretical accordance with the observations reported in \cite{Olsson} concerning the LuGre model: the ODE~\eqref{eq:ODEzGen} describes the average behavior of the bristle elements, and is thus not capable of capturing instantaneous sticking effects. Hence, the mathematical derivation of~\eqref{eq:ODEzGen} provides a formal (although not mathematically rigorous) justification for the informal arguments contained in \cite{Olsson}.
\end{remark}


\section{A LuGre-like lumped FrBD model}\label{eq:LumpedLuGre}
The present section is dedicated to the derivation and analysis of a lumped, first-order FrBD model that closely resembles the LuGre one (albeit possessing some peculiar features that significantly affect its dissipative behavior). In particular, its governing equations are established in Section~\ref{sect:modelEq}. Some important mathematical properties, intimately related to the physical interpretation of the problem, are then analyzed in Section~\ref{sect:MathematicalProp}. Finally, a linearized version of the model is delivered in Section~\ref{sect:linearLump}.

\subsection{Model equations and solution}\label{sect:modelEq}
This section introduces the governing equations of the novel LuGre-like FrBD model, and discusses the salient properties of its solution. Specifically, the mathematical derivation of the model is conducted in Section~\ref{sect:der}, whereas Section~\ref{ect:wellP} provides both qualitative and quantitative results concerning its well-posedness.

\subsubsection{Model equations}\label{sect:der}

The specific structure of the ODE~\eqref{eq:ODEzGen} may be deduced by opportunely postulating an analytical expression for the bristle force $F\ped{b} = f(\dot{z},z,v)$. In this paper, a viscoelastic constitutive relationship is assumed for the bristle element, so that the corresponding tangential force may be deduced as
\begin{align}\label{eq:fbrist}
F\ped{b} = f(\dot{z},z,v) =\mu\ped{b}(\dot{z},z,v)p =  \bigl(\sigma_0(v)z + \sigma_1(v) \dot{z}\bigr)p, 
\end{align}
where $\sigma_0 \in C^1(\mathbb{R};\mathbb{R}_{\geq 0})$ denotes the \emph{normalized micro-stiffness}, $\sigma_1 \in C^1(\mathbb{R}; \mathbb{R}_{\geq 0})$ the \emph{normalized micro-damping}, and $\mu\ped{b} \in C^1(\mathbb{R}^3;\mathbb{R})$ a \emph{virtual friction coefficient}, defined in analogy with the LuGre model \cite{LuGreDistr}. It is worth observing that~\eqref{eq:fbrist} does not contain the viscous contribution that instead appears in the LuGre formulation. In fact, according to the analysis conducted in Section~\ref{sect:impl}, the viscous term should be more realistically accounted for in the definition of the friction coefficient \cite{Rill,Trib1}.

Consequently, proceeding as explained in Section~\ref{sect:impl} yields the following ODE for the bristle dynamics:
\begin{subequations}\label{eq:ODEz}
\begin{align}
& \dot{z}(t) = -\dfrac{\sigma_0(v)\abs{v(t)}_\varepsilon}{g\bigl(v(t)\bigr)}z(t) +\dfrac{\mu\bigl(v(t)\bigr)}{g\bigl(v(t)\bigr)}v(t), \quad  t\in(0,T),\label{eq:ODEz0} \\
& z(0) = z_0,
\end{align}
\end{subequations}
where the function $g(\cdot)$ in~\eqref{eq:g(v)Gen} reads explicitly
\begin{align}\label{eq:g(v)}
g(v) \triangleq \sigma_1(v)\abs{v}_\varepsilon +\mu(v).
\end{align}
Apart from the viscous contribution already included in~\eqref{eq:fViscOOO},~\eqref{eq:ODEz} and~\eqref{eq:g(v)} differ from the LuGre formulation by incorporating the additional damping-related term $\sigma_1(v)\abs{v}_{\varepsilon}$ into~\eqref{eq:g(v)}. 

Moreover, following \cite{LuGreDistr}, the virtual friction coefficient may be restated more conveniently as
\begin{align}
\begin{split}
\mu\ped{b}(t) & = \sigma_0\bigl(v(t)\bigr)z(t) + \sigma_1\bigl(v(t)\bigr)\dot{z}(t) \\
& = \bar{\sigma}_0\bigl(v(t)\bigr)z(t) + \bar{\sigma}_2\bigl(v(t)\bigr)v(t),
\end{split}
\end{align}
with
\begin{subequations}\label{eq:mu}
\begin{align}
\bar{\sigma}_0(v) &\triangleq \sigma_0(v)\biggl( 1-\dfrac{\sigma_1(v)\abs{v}_\varepsilon}{g(v)}\biggr), \\
\bar{\sigma}_2(v) &\triangleq \sigma_1(v)\dfrac{\mu(v)}{g(v)}.
\end{align}
\end{subequations}
The output of the ODE system~\eqref{eq:ODEz} may be finally defined as
\begin{align}\label{eq:outputLump}
F\ped{b}(t) = \mu\ped{b}(t)p.
\end{align}

Before proceeding with the model analysis, a consideration is reported below.
\begin{remark}
The coefficients $\sigma_0(v)$ and $\sigma_1(v)$ in~\eqref{eq:fbrist} have been assumed to depend on the relative velocity $v$; this is in analogy with the LuGre model, where the micro-damping term is modeled as a function of $v$ to ensure passivity. However, nonlinear stiffness and damping coefficients may be more generally postulated as functions of $\dot{z}$, $z$, and $v$. For instance, if $\sigma_1(\dot{z},v)$ satisfies $\sigma_1(0,v) = 0$, then the FrBD model~\eqref{eq:ODEz}-\eqref{eq:outputLump} reduces exactly to the LuGre one, but with a nonlinear micro-damping coefficient that depends also on $\dot{z}$, in addition to $v$. Of course, both the FrBD description~\eqref{eq:ODEz}-\eqref{eq:outputLump} and the LuGre one become formally equivalent to the Dahl model for $\sigma_1(\dot{z},v) = 0$.
\end{remark}

\subsubsection{Well-posedness and solution}\label{ect:wellP}
The FrDB model~\eqref{eq:ODEz}-\eqref{eq:outputLump} comprises a nonlinear ODE, whose well-posedness may be established using standard arguments. For the sake of completeness, existence and uniqueness results for the solution of~\eqref{eq:ODEz} are collected in Theorem~\ref{thm:Mild0} below.

\begin{theorem}[Existence and uniqueness of solutions]\label{thm:Mild0}
For all inputs $v \in C^0([0,T];\mathbb{R})$ and ICs $\mathbb{R} \ni z_0 \triangleq z(0)$, the ODE~\eqref{eq:ODEz} admits a unique solution $z \in C^1([0,T];\mathbb{R})$.

\begin{proof}
The result immediately follows from Theorem 2.3 in \cite{Khalil} (see, e.g., Property 3.1 in \cite{Olsson}).
\end{proof}
\end{theorem}
\begin{remark}[Global existence and uniqueness]\label{remark:1jdee}
For inputs $v \in C^0(\mathbb{R}_{\geq 0};\mathbb{R})\cap L^\infty(\mathbb{R}_{\geq 0};\mathbb{R})$, the solution obtained from Theorem~\ref{thm:Mild0} is global, and the interval $[0,T]$ may be taken as $\mathbb{R}_{\geq 0}$ (i.e., $T = \infty$).
\end{remark}

Theorem~\ref{thm:Mild0} and Remark~\ref{remark:1jdee} consider the ODE~\eqref{eq:ODEz} in isolation. An important consideration concerning its possible interconnection with other ODE systems is formalized below.
\begin{remark}
Compared to the standard LuGre formulation, modeling the functions $\mu(\cdot)$ and $g(\cdot)$ according to~\eqref{eq:fViscOOO} and~\eqref{eq:g(v)} may have important consequences on the establishment of existence and uniqueness results for the equations governing the dynamics of mechanical systems with friction. For instance, for constant $\sigma_0(v) = \sigma_0$ and $\sigma_1(v) = \sigma_1$, the term $\sigma_0\abs{v}_\varepsilon/g(v)$ in~\eqref{eq:ODEz} becomes uniformly bounded, which could ensure well-posedness where the LuGre model would normally fail. This modeling approach also ensures passivity virtually for every combination of parametrizations (Lemma~\ref{lemma:DissF}), as opposed to the LuGre formulation.
\end{remark}

Whereas Theorem~\ref{thm:Mild0}only provides a qualitative result, the solution to~\eqref{eq:ODEz} may be recovered explicitly regarding the ODE~\eqref{eq:ODEz0} as a linear time-varying system. Accordingly, 
\begin{align}
z(t) = \Phi(t,0)z_0 + \int_0^t \Phi\bigl(t,t^\prime)\dfrac{\mu\bigl(v(t^\prime)\bigr)}{g\bigl(v(t^\prime)\bigr)}v\bigl(t^\prime\bigr) \dif t^\prime, \quad t \in [0,T],
\end{align}
with
\begin{align}
\Phi(t,\tilde{t}) & \triangleq \exp\Biggl(-\int_{\tilde{t}}^t \dfrac{\sigma_0\bigl(v(t^\prime)\bigr)\abs{v(t^\prime)}_\varepsilon}{g\bigl(v(t^\prime)\bigr)}\dif t^\prime\Biggr).
\end{align}
The corresponding stationary solution reads obviously
\begin{align}\label{eq:ssSolLump}
z = \sgn_\varepsilon(v)\dfrac{\mu(v)}{\sigma_0(v)},
\end{align}
and coincides with that of the original equation~\eqref{eq:nonlinImplGen}. Not surprisingly, it also coincides -- at least formally -- with the stationary solution of the LuGre friction model, which is intelligently constructed \emph{ad hoc} to yield an identical expression to~\eqref{eq:ssSolLump} in steady-state conditions. In the present case, however, the formula~\eqref{eq:ssSolLump} has been deduced \emph{a posteriori}, and is grounded on a clear, first-order approximation of a physical model for the bristle dynamics. Starting with~\eqref{eq:nonlinImplGen}, the stationary force may also be inferred to match the friction characteristic:
\begin{align}\label{eq:ssFOrceLump}
F\ped{b} = \sgn_\varepsilon(v)\mu(v)p = F\ped{r}\bigl(v\ped{s}(0,v)\bigr),
\end{align}
which was expected from~\eqref{eq:fbrist} and~\eqref{eq:ffric}. 


\subsection{Mathematical properties}\label{sect:MathematicalProp}

The present section investigates some important mathematical features of the FrDB model~\eqref{eq:ODEz}-\eqref{eq:outputLump}, including \emph{stability} and \emph{passivity}. These are analyzed in Sections~\ref{sect:stability} and~\ref{sect:diss}, respectively.

\subsubsection{Stability}\label{sect:stability}
Intuitively, the bristle deflection $z(t)$ should be bounded when the two contacting bodies are subjected to finite relative velocities. Lemma~\ref{lemma:stab} formalizes this concept from a mathematical viewpoint.

\begin{lemma}[Stability]\label{lemma:stab}
Suppose that $C^0(\mathbb{R};[\bar{\mu}\ped{min}, \bar{\mu}\ped{max}]) \ni \bar{\mu} \triangleq \dfrac{\mu}{\sigma_0}$, with $\bar{\mu}\ped{min}, \bar{\mu}\ped{max} \in \mathbb{R}_{>0}$, and that $\abs{\sgn_\varepsilon(\cdot)} \leq 1$. Then, for all inputs $v\in C^0(\mathbb{R}_{\geq 0};\mathbb{R})$ and ICs $\abs{z_0}\leq \bar{\mu}\ped{max} $, the ODE~\eqref{eq:ODEz} satisfies $\abs{z(t)}\leq \bar{\mu}\ped{max} $ for all $t\in \mathbb{R}_{\geq 0}$.

\begin{proof}
See Property 3.2 in \cite{Olsson}.
\end{proof}
\end{lemma}
Clearly, under the same assumptions of Lemma~\ref{lemma:stab}, it is also possible to prove the stability of the output $F\ped{b}(t)$. The fact that bounded deflections produce bounded frictional forces is eventually trivial for the lumped version of the LuGre-like FrDB model, and therefore not formally investigated in the present paper. As mentioned in Section~\ref{sect:stability}, the same implication is not necessarily valid for its distributed counterpart.

\subsubsection{Dissipativity and passivity}\label{sect:diss}
When it comes to a friction model, two extremely desirable mathematical properties are \emph{dissipativity} and \emph{passivity}. In particular, the latter intimately pertains to the physical domain: its interpretation is that the frictional force generated by the bristle dynamics should dissipate energy over time. Concerning the FrDB model described by~\eqref{eq:ODEz}-\eqref{eq:outputLump} (but more generally~\eqref{eq:FbRate} and~\eqref{eq:ODEzGen}-\eqref{eq:g(v)Gen}), the notions of dissipativity and passivity are enounced according to Definition~\ref{def:dissip}. 
\begin{definition}[Dissipativity and passivity]\label{def:dissip}
The system~\eqref{eq:ODEz}-\eqref{eq:outputLump} is called \emph{dissipative} if, for all inputs $v \in C^0([0,T];\mathbb{R})$ and ICs $z_0 \in \mathbb{R}$, there exist a supply rate $w : \mathbb{R}^2\mapsto \mathbb{R}$ and storage function $W : \mathbb{R} \mapsto \mathbb{R}_{\geq 0}$ such that
\begin{align}\label{eq:Fvres0P}
\begin{split}
\int_0^t w\bigl(F\ped{b}(t^\prime),v(t^\prime)\bigr) \dif t^\prime \geq W\bigl(z(t)\bigr)& - W\bigl(z_0\bigr), \quad t \in [0,T].
\end{split}
\end{align}
It is called \emph{passive} if $w(F\ped{b}(t),v(t)) = F\ped{b}(t)v(t)$.
\end{definition}

\begin{lemma}[Passivity]\label{lemma:DissF}
Suppose that $\sigma_0(v) = \sigma_0 \in \mathbb{R}_{>0}$.
Then, the FrBD model~\eqref{eq:ODEz}-\eqref{eq:outputLump} is passive with storage function
\begin{align}\label{eq:VdissF0}
W\bigl(z(t)\bigr) \triangleq \dfrac{\Sigma_0}{2}z^2(t),
\end{align}
where the \emph{velocity-independent micro-stiffness} coefficient is defined as $\mathbb{R}_{>0} \ni \Sigma_0 \triangleq \sigma_0 p$.

\begin{proof}
Using the ODE~\eqref{eq:ODEz} provides
\begin{align}\label{eq:firstBound}
\begin{split}
F\ped{b}(t)v(t) & = \Sigma_0\Biggl(1 - \dfrac{\sigma_1\bigl(v(t)\bigr)\abs{v(t)}_\varepsilon}{g\bigl(v(t)\bigr)}\Biggr)z(t)v(t) \\
& \quad + p\sigma_1\bigl(v(t)\bigr)\dfrac{\mu\bigl(v(t)\bigr)}{g\bigl(v(t)\bigr)}v^2(t).
\end{split}
\end{align}
Adding and subtracting $\dot{W}(z(t))$ as in~\eqref{eq:VdissF0}, and recalling~\eqref{eq:g(v)}, yields
\begin{align}\label{eq:firstBound2}
\begin{split}
& F\ped{b}(t)v(t) -\dot{W}\bigl(z(t)\bigr) = \Sigma_0\sigma_0\dfrac{\abs{v(t)}_\varepsilon}{g\bigl(v(t)\bigr)}z^2(t) \\
& \qquad + p\sigma_1\bigl(v(t)\bigr)\dfrac{\mu\bigl(v(t)\bigr)}{g\bigl(v(t)\bigr)}v^2(t), \quad t \in (0,T).
\end{split}
\end{align}
Hence,
\begin{align}\label{eq:diss2}
F\ped{b}(t)v(t)  \geq \dot{W}\bigl(z(t)\bigr), \quad t \in (0,T).
\end{align}
Integrating the above~\eqref{eq:diss2} yields~\eqref{eq:Fvres0P} with $W(z(t))$ defined according to~\eqref{eq:VdissF0}.
\end{proof}
\end{lemma}
Provided that the micro-stiffness coefficient is constant, Lemma~\ref{lemma:DissF} establishes the passivity of~\eqref{eq:ODEz}-\eqref{eq:outputLump} for any admissible parametrization. This is a modeling advantage compared to the LuGre formulation, where passivity requires postulating a velocity-dependent micro-damping, which appears to be a rather artificial modification. 

The next Section~\ref{sect:linearLump} introduces two linearized versions of equations~\eqref{eq:ODEz}-\eqref{eq:outputLump}. 

\subsection{Linearized model}\label{sect:linearLump}

The present section is dedicated to the derivation of a linearized version of the FrBD model~\eqref{eq:ODEz}-\eqref{eq:outputLump}. To this end, the input and states are decomposed respectively as $v(t) = v^\star + \tilde{v}(t)$ and $z(t) = z^\star + \tilde{z}(t)$, where $v^\star$, $z^\star \in \mathbb{R}$ represent a constant input and the corresponding stationary solution, and $\tilde{v}(t)$, $\tilde{z}(t)\in \mathbb{R}$ indicate small perturbations around these. Accordingly, linearization of~\eqref{eq:PDEfric} yields
\begin{subequations}\label{eq:LuGreDistrDelta0W}
\begin{align}
\begin{split}
& \dot{\tilde{z}}(t) = -\dfrac{\sigma_0(v^\star)\abs{v^\star}_\varepsilon}{g(v^\star)}\tilde{z}(t) + H_1(z^\star,v^\star)\tilde{v}(t), \quad t \in (0,T), \label{eq:LugreODEDelta0}
\end{split} \\
& \tilde{z}(0) = \tilde{z}_{0}, 
\end{align}
\end{subequations}
and the corresponding linearized friction coefficient reads
\begin{align}\label{eq:mu_delta0}
\tilde{\mu}\ped{b}(t) \triangleq \bar{\sigma}_0(v^\star)\tilde{z}(t) + H_2(z^\star,v^\star)\tilde{v}(t),
\end{align}
where \begin{small}
\begin{subequations}
\begin{align}
\begin{split}
 H_1(z,v) & \triangleq \dfrac{\mu(v)}{g(v)}\biggl(1-\dfrac{v}{g(v)}\dod{g(v)}{v}\biggr)+\dfrac{v}{g(v)}\dod{\mu(v)}{v}\\
& \quad - \dfrac{\sigma_0(v)z}{g(v)}\biggr(\dod{\abs{v}_\varepsilon}{v} - \dfrac{\abs{v}_\varepsilon}{g(v)}\dod{g(v)}{v}\biggl) -\dod{\sigma_0(v)}{v}\dfrac{\abs{v}_\varepsilon}{g(v)}, 
\end{split}\\
\begin{split}
H_2(z,v) & \triangleq \dod{\bar{\sigma}_0(v)}{v}z + \dod{\bar{\sigma}_2(v)}{v}v + \bar{\sigma}_2(v),
\end{split}
\end{align}
\end{subequations} \end{small}
and \begin{small}
\begin{subequations}
\begin{align}
\begin{split}
 \dod{\bar{\sigma}_0(v)}{v} & = \dod{\sigma_0(v)}{v}\biggl( 1-\dfrac{\sigma_1(v)\abs{v}_\varepsilon}{g(v)}\biggr)-\dfrac{\sigma_0(v)\abs{v}_\varepsilon}{g(v)}\dod{\sigma_1(v)}{v}\\
& \quad-\dfrac{\sigma_0(v)\sigma_1(v)}{g(v)}\biggl(\dod{\abs{v}_\varepsilon}{v}-\dfrac{\abs{v}_\varepsilon}{g(v)}\dod{g(v)}{v}\biggr), 
\end{split} \\
\begin{split}
\dod{\bar{\sigma}_2(v)}{v} & = \dfrac{\sigma_1(v)}{g^2(v)}\biggl(g(v)\dod{\mu(v)}{v}-\mu(v)\dod{g(v)}{v}\biggr) + \dfrac{\mu(v)}{g(v)}\dod{\sigma_1(v)}{v}.
\end{split}
\end{align}
\end{subequations} \end{small}
The corresponding linearized force is formally given by
\begin{align}\label{eq:Flin10}
\tilde{F}\ped{b}(t) = \tilde{\mu}\ped{b}(t)p.
\end{align}
By visual inspection, it is straightforward to conclude that, as for its nonlinear counterpart, the linearized model described by~\eqref{eq:LuGreDistrDelta0W}-\eqref{eq:Flin10} also enjoys nice stability properties. This concludes the analysis of the lumped version of the FrBD model. The next Section~\ref{sect:Distr} explores its extension to the distributed case.

\section{A LuGre-like distributed FrBD model}\label{sect:Distr}
When modeling rolling contact systems, it is necessary to consider the distributed nature of the frictional phenomena that govern rolling and sliding motions. This section is thus dedicated to extending the LuGre-like FrBD model introduced in Section~\ref{eq:LumpedLuGre} to the distributed case. Such a generalization is achieved in a relatively straightforward manner, but requires more sophisticated mathematical tools to be handled. Specifically, the governing equations of the distributed variant are detailed in Section~\ref{sect:MESdistr}, whereas the notions of stability and passivity are formalized and investigated mathematically in Section~\ref{sect:NMEPdistr}. Finally, in a direct analogy to what done previously, a linearized version of the distributed FrBD model is delivered in Section~\ref{sect:LinearDsitr}.

\subsection{Model equations and solution}\label{sect:MESdistr}
The governing equations of the distributed LuGre-like FrBD model are introduced in Section~\ref{sect:DerDistr}, whilst well-posedness properties are briefly discussed in Section~\ref{sect:WellpDistr}, where a closed-form solution is also derived explicitly.

\subsubsection{Model equations}\label{sect:DerDistr}

\begin{figure}
\centering
\includegraphics[width=0.8\linewidth]{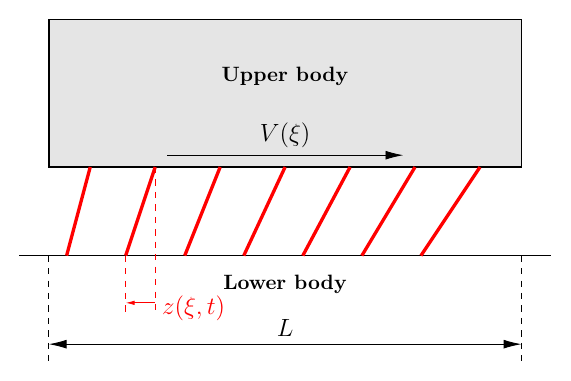} 
\caption{A schematic representation of the distributed FrBD friction model.}
\label{fig:DistrModel}
\end{figure}

When the frictional rolling contact is distributed inside a region of length $L \in \mathbb{R}_{>0}$, the ODE~\eqref{eq:ODEzGen} becomes a PDE. Such a PDE may be derived directly from~\eqref{eq:ODEzGen} by replacing the total Lagrangian derivative with the Eulerian one:
\begin{align}\label{eq:ODEEEEE}
\dod{z(\xi,t)}{t} = \dpd{z(\xi,t)}{t} + V(\xi)\dpd{z(\xi,t)}{\xi},
\end{align}
where $z(\xi,t)\in \mathbb{R}$ represents the distributed state, and $V \in C^1([0,1];\mathbb{R}_{>0})$ denotes the transport velocity. This approach is very general and permits recovering a family of distributed FrBD models from their lumped counterparts. It is worth noting that, if $V(\xi) = 0$,~\eqref{eq:ODEEEEE} in conjunction with~\eqref{eq:ODEzGen} yields again an ODE, where the variable $\xi \in [0,L]$ may be regarded as a parameter. Similarly, all model coefficients may also be treated as spatially varying. A schematic of the distributed friction model is illustrated in Figure~\ref{fig:DistrModel}. 

Considering specifically the (possibly regularized) LuGre-like FrBD model with bristle dynamics~\eqref{eq:ODEz}, the following semilinear PDE, rescaled on a unitary domain, may be deduced for the distributed case:
\begin{subequations}\label{eq:PDEfric}
\begin{align}
\begin{split}
&\dpd{z(\xi,t)}{t} + V(\xi)\dpd{z(\xi,t)}{\xi}  = -\dfrac{\sigma_0\bigl(v(t),\xi\bigr)\abs{v(t)}_\varepsilon}{g\bigl(v(t),\xi\bigr)}z(\xi,t) \\
&\qquad \qquad  \qquad +\dfrac{\mu\bigl(v(t)\bigr)}{g\bigl(v(t),\xi\bigr)}v(t),\quad (\xi,t) \in(0,1)\times(0,T),
\end{split}\label{eq:PDEfric1}\\
 & z(0,t) = 0,\quad  t\in(0,T), \label{eq:PDEfricBC}\\
& z(\xi,0) = z_0(\xi), \quad \xi \in (0,1). \label{eq:PDEfricIC}
\end{align}
\end{subequations}
with
\begin{align}\label{eq:gVar}
g(v, \xi) \triangleq \sigma_1(v,\xi)\abs{v}_\varepsilon +\mu(v).
\end{align}
Equation~\eqref{eq:PDEfric} may be used to describe the relative deformation between two rolling contact bodies, such as tire and road, wheel and rail, but also bearings and machine elements. 
In~\eqref{eq:PDEfric} and~\eqref{eq:gVar}, $\sigma_0 \in C^1(\mathbb{R}\times[0,1];\mathbb{R}_{\geq 0})$ denotes the spatially varying normalized micro-stiffness coefficient, and $\sigma_1 \in C^1(\mathbb{R}\times[0,1];\mathbb{R}_{\geq 0})$ the spatially-varying normalized micro-damping coefficient; for notational convenience, the input $v(t)$ is supposed to be independent of the spatial coordinate\footnote{As also explained in Section~\ref{sect:diss}, none of the paper's conclusion is significantly affected by this assumption, whereas the dependency of other model parameters on the spatial coordinate may have profound implications.}, albeit the generalization to a spatially varying velocity is straightforward. The boundary condition (BC)~\eqref{eq:PDEfricBC} prescribes that the bristles should enter the contact area in an undeformed state: this is a standard assumption in the theory of elastic rolling contact (where it derives from the continuity of the stresses), but has been historically extended also to viscoelastic rolling contact problems, as testified by many contributions dealing with the distributed LuGre formulation. 

Moreover, exactly as for its lumped counterpart, the virtual friction coefficient may be written as
\begin{align}\label{eq:muVDistr}
\begin{split}
\mu\ped{b}(\xi,t) & = \sigma_0\bigl(v(t),\xi\bigr)z(\xi,t) + \sigma_1\bigl(v(t),\xi\bigr)\dod{z(\xi,t)}{t} \\
& = \bar{\sigma}_0\bigl(v(t),\xi\bigr)z(\xi,t) + \bar{\sigma}_2\bigl(v(t),\xi\bigr)v(t),
\end{split}
\end{align}
with
\begin{subequations}\label{eq:muXi}
\begin{align}
\bar{\sigma}_0(v,\xi) &\triangleq \sigma_0(v,\xi)\biggl( 1-\sigma_1(\xi)\dfrac{\abs{v}_\varepsilon}{g(v)}\biggr), \\
\bar{\sigma}_2(v,\xi) &\triangleq \sigma_1(v,\xi)\dfrac{\mu(v)}{g(v)}.
\end{align}
\end{subequations} 
Denoting the distributed output, that is, the infinitesimal bristle force, by
\begin{align}\label{eq:outputDistr22}
f\ped{b}(\xi,t) = L\mu\ped{b}(\xi,t)p(\xi),
\end{align}
where $p \in C^1([0,1];\mathbb{R}_{\geq 0})$ varies now over space, the finite-dimensional output of the PDE~\eqref{eq:PDEfric} may be finally defined as
\begin{align}\label{eq:outputDistr}
F\ped{b}(t) = \bigl(\mathscr{K}(z,v)\bigr)(t) = \int_0^1f\ped{b}(\xi,t)\dif \xi.
\end{align}
For what follows, it may be beneficial to identify the output of the PDE~\eqref{eq:PDEfric} as~\eqref{eq:outputDistr} or~\eqref{eq:outputDistr22}, depending on the application.

\begin{remark}
It is crucial to observe that the total time derivative is used in the definition of the virtual friction coefficient in~\eqref{eq:muVDistr}. Replacing the total time derivative with the partial one, as done by some authors for the distributed LuGre model, might introduce inconsistencies concerning the stability and dissipative behavior of the system~\eqref{eq:PDEfric}-\eqref{eq:outputDistr}, as extensively discussed in \cite{LuGreDistr}. Moreover, the partial derivative does not have a clear physical meaning, whereas the total one represents a real deformation velocity. It is also worth emphasizing that, if the total derivative is employed in~\eqref{eq:muVDistr}, for $v(t) \in \mathbb{R}$ fixed, $F\ped{b}(t) = (\mathscr{K}(z,v))(t)$ may be interpreted as a bounded operator on $L^2((0,1);\mathbb{R})$, whereas using the partial derivative would make it unbounded. 
\end{remark}

\subsubsection{Well-posedness and solution}\label{sect:WellpDistr}
Well-posedness for the PDE~\eqref{eq:PDEfric} may be proved within different functional settings. Following \cite{LuGreDistr}, the present paper focuses on the notions of \emph{classical} and \emph{mild solutions}. Existence and uniqueness results are asserted by Theorems~\ref{thm:Class} and~\ref{thm:Mild}.


\begin{theorem}[Existence and uniqueness of classical solutions]\label{thm:Class}
Suppose that $\varepsilon \in \mathbb{R}_{>0}$ in~\eqref{eq:PDEfric}. Then, for all inputs $v \in C^1([0,T];\mathbb{R})$ and initial conditions (ICs) $z_0 \in H^1((0,1);\mathbb{R})$ satisfying the BC~\eqref{eq:PDEfricBC}, the distributed friction model~\eqref{eq:PDEfric} admits a unique \emph{classical solution} $z \in C^1([0,T];L^2((0,1);\mathbb{R})) \cap C^0([0,T];H^1((0,1);\mathbb{R}))$ satisfying the BC~\eqref{eq:PDEfricBC}.

\begin{proof}
The result follows from an application of Theorem 6.1.5 in \cite{Pazy} (see also Theorem 2.1 in \cite{LuGreDistr} for additional details).
\end{proof}
\end{theorem}

\begin{theorem}[Existence and uniqueness of mild solutions]\label{thm:Mild}
For all inputs $v \in C^0([0,T];\mathbb{R})$ and initial conditions (ICs) $z_0 \in L^2((0,1);\mathbb{R})$, the distributed friction model~\eqref{eq:PDEfric} admits a unique \emph{mild solution} $z \in C^0([0,T];L^2((0,1);\mathbb{R}))$.

\begin{proof}
The result follows from an application of Theorem 6.1.2 in \cite{Pazy}.
\end{proof}
\end{theorem}
\begin{remark}[Global existence and uniqueness]
For inputs $v \in C^1(\mathbb{R}_{\geq 0};\mathbb{R})\cap L^\infty(\mathbb{R}_{\geq 0};\mathbb{R})$ and $v \in C^0(\mathbb{R}_{\geq 0};\mathbb{R})\cap L^\infty(\mathbb{R}_{\geq 0};\mathbb{R})$, respectively, the solutions obtained from Theorems~\ref{thm:Class} and~\ref{thm:Mild} are global, and the interval $[0,T]$ may be taken as $\mathbb{R}_{\geq 0}$, i.e., $T = \infty$ (see Theorem 6.1.2 in \cite{Pazy} and the accompanying commentary).
\end{remark}

Whereas the result asserted by Theorems~\ref{thm:Class} and~\ref{thm:Mild} are merely qualitative, sufficiently smooth solutions to the PDE~\eqref{eq:PDEfric} (such as classical) may be recovered in closed form using similar techniques as those employed in \cite{LuGreDistr,Aamo}. 
Indeed, the method of the characteristic lines gives
\begin{align}\label{eq:zSolEta}
\begin{split}
 & z(\xi,t)  = \int_{\max(\varpi(\xi)-t,0)}^{\varpi(\xi)} \Phi\bigl(\varpi(\xi),\eta,t\bigr) \\
& \quad \times \dfrac{\mu\bigl(\eta-\varpi(\xi)+t\bigr)}{g\bigl(\eta-\varpi(\xi)+t, \varpi^{-1}(\eta)\bigr)}v\bigl(\eta-\varpi(\xi)+t\bigr)\dif \eta \\
&\quad + \Phi\bigl(\varpi(\xi),\varpi(\xi)-t,t\bigr) z_0\Bigl(\varpi^{-1}\bigl(\max(\varpi(\xi)-t,0)\bigr)\Bigr), \\
& \qquad\qquad\qquad\qquad\qquad\qquad\quad  (\xi,t) \in [0,1]\times[0,T],
\end{split}
\end{align}
where the mapping $\varpi : [0,1]\mapsto \mathbb{R}_{\geq 0}$, defined by
\begin{align}\label{eq:PiXi}
 \varpi(\xi) \triangleq \int_0^\xi \dfrac{1}{V(\xi^\prime)}\dif \xi^\prime,
\end{align}
is monotonically increasing and invertible since $V \in C^1([0,1];\mathbb{R}_{>0})$, and
\begin{align}\label{eq:PhiTrans}
\begin{split}
\Phi(\eta,\tilde{\eta},t) & \triangleq \exp\Biggl( -\int_{\tilde{\eta}}^\eta\sigma_0\bigl(v(\eta^\prime-\eta +t),\varpi^{-1}(\eta^\prime)\bigr) \\
& \quad \times \dfrac{\abs{v(\eta^\prime-\eta +t)}_\varepsilon}{g\bigl(v(\eta^\prime-\eta +t),\varpi^{-1}(\eta^\prime)\bigr)} \dif \eta^\prime \Biggr).
\end{split}
\end{align} 
Equation~\eqref{eq:zSolEta} represents the most general expression for the bristle deflection in transient conditions.
In particular, for $V(\xi) = V$, $\sigma_0(v,\xi) = \sigma_0(v)$, and $\sigma_1(v,\xi) = \sigma_1(v)$ constant, the integrals in~\eqref{eq:zSolEta} and~\eqref{eq:PhiTrans} considerably simplify, and the stationary solution assumes the form
\begin{align}\label{eq:CD2}
z(\xi) = \sgn_\varepsilon(v)\dfrac{\mu(v)}{\sigma_0(v)}\Biggl[1-\exp\biggl( -\dfrac{\sigma_0(v)\abs{v}_\varepsilon}{V g(v)}\xi\biggr)\Biggr], \quad \xi \in [0,1],
\end{align}
which is consistent with the expression deduced according to the distributed LuGre model (with $\sigma_0(v) = \sigma_0$ and $g(v) = \mu(v)$). 

Starting with~\eqref{eq:CD2}, an analytical expression for the total frictional force may be deduced by opportunely specifying the function $p(\cdot)$.
In the literature, different shapes for the pressure distribution have been postulated. Amongst those more frequently used, there is the constant one $p(\xi) = p_0$, the exponentially decreasing one $p(\xi) = p_0\exp(-a\xi)$, and the parabolic one $p(\xi) = p_0\xi(1-\xi)$.
In particular, assuming a constant pressure distribution gives 
\begin{small}
\begin{align}\label{eq:Fss1}
\begin{split}
F\ped{b} &= \dfrac{vLp_0V\bar{\sigma}_0(v)\mu(v)g(v)}{\sigma_0^2(v) \abs{v}_\varepsilon^2} \Biggl[\dfrac{\sigma_0(v)\abs{v}_\varepsilon}{V g(v)}+\exp\biggl( -\dfrac{\sigma_0(v)\abs{v}_\varepsilon}{V g(v)}\biggr)-1\Biggr] \\
& \quad  + Lp_0\bar{\sigma}_2(v)v,
\end{split}
\end{align}\end{small}
whereas the exponential distribution yields \begin{small}
\begin{align}\label{eq:Fss2}
\begin{split}
F\ped{b} &= \dfrac{\sgn_\varepsilon(v)Lp_0V\bar{\sigma}_0(v)\mu(v)g(v)}{\sigma_0(v)\bigl(aV g(v)+\sigma_0(v)\abs{v}_\varepsilon\bigr)} \Biggl[\exp\biggl( -\dfrac{aV g(v)+\sigma_0(v)\abs{v}_\varepsilon}{V g(v)}\biggr)-1\Biggr] \\
& \quad + \dfrac{Lp_0\bigl(1-\exp(-a)\bigr)}{a}\biggl(\sgn_\varepsilon(v)\dfrac{\bar{\sigma}_0(v)\mu(v)}{\sigma_0(v)}+ \bar{\sigma}_2(v)v\biggr).
\end{split}
\end{align}\end{small}
Once again, both the expressions~\eqref{eq:Fss1} and~\eqref{eq:Fss2} are in theoretical agreement with the corresponding ones deduced using the distributed LuGre model. The relationship between the normalized steady-state bristle force $\frac{F\ped{b}}{F_z}$ and the relative velocity $\frac{v}{LV}$ (also known as \emph{slip}), where $F_z \in \mathbb{R}_{>0}$ denotes the total vertical force acting on the body, is visualized in Figure~\ref{fig:Fb}, where data collected from a tire operating in a low-friction environment \cite{Two-regime} is also displayed. From Figure~\ref{fig:Fb}, it may be concluded that the distributed model can describe with high precision the steady-state tire force as a function of the slip $\frac{v}{LV}$. 
\begin{figure}
\centering
\includegraphics[width=1\linewidth]{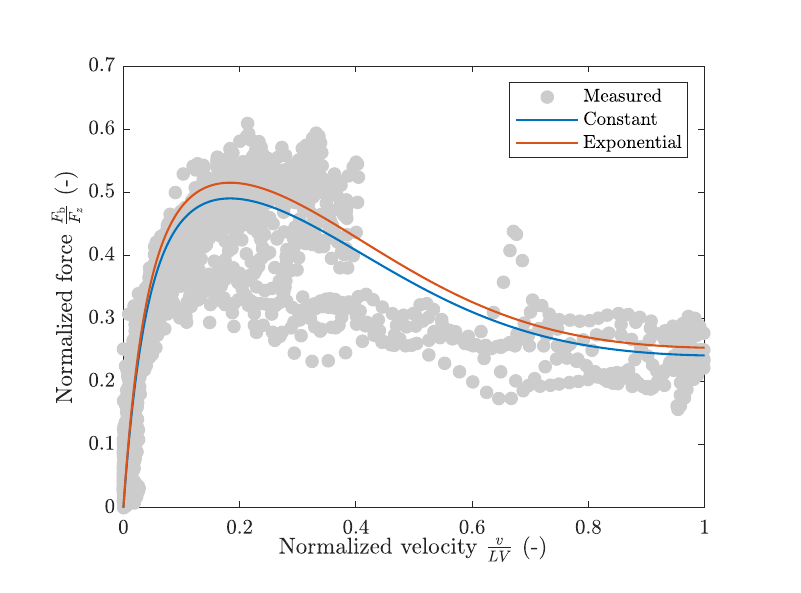} 
\caption{Normalized steady-state bristle force $\frac{F\ped{b}}{F_z}$ for the constant pressure distribution (blue line), and exponential one (orange line), for a tire operating in a low-friction environment. Model parameters as in Table~\ref{tab:param1}.}
\label{fig:Fb}
\end{figure}

Before moving to the mathematical analysis of the distributed FrBD model, a final consideration is summarized in Remark~\ref{Remark:Vs} below.

\begin{remark}\label{Remark:Vs}
In many systems of practical interest, the transport velocity does not depend on the spatial coordinate, but only on the time variable, i.e., $V = V(t)$. In such cases, owing to the assumption $V \in C^1([0,T];\mathbb{R}_{>0})$, it is possible to define a new time-like variable $\mathbb{R}_{\geq 0}\ni s\triangleq \int_0^t V(t^\prime) \dif t^\prime$ and recast the PDE~\eqref{eq:PDEfric} as a linear time-varying system with constant transport velocity. For the resulting equation, existence and uniqueness results would follow easily from similar arguments as those adopted in the proof of Theorem 2.1 in \cite{LuGreDistr}. For instance, concerning the distributed LuGre model, explicit solutions have been obtained in \cite{CarcassDyn} resorting to analogous strategies to those adopted for the derivation of~\eqref{eq:zSolEta}.
\end{remark}

\begin{table}[h!]\centering 
\caption{Model parameters}
{\begin{tabular}{|c|c|c|c|}
\hline
Parameter & Description & Unit & Value \\
\hline 
$L$ & Contact length & m & 0.1 \\
$\sigma_0$ & Normalized micro-stiffness & $\textnormal{m}^{-1}$ & 252 \\
$\sigma_1$ & Normalized micro-damping & $\textnormal{s}\,\textnormal{m}^{-1}$ &0 \\
$\sigma_2$ & Normalized viscous damping & $\textnormal{s}\,\textnormal{m}^{-1}$ &0.0018 \\
$\mu\ped{d}$ & Dynamic friction coefficient & - &0.2 \\
$\mu\ped{s}$ & Static friction coefficient & - &0.6 \\
$v\ped{S}$ & Stribeck velocity & $\textnormal{m}\,\textnormal{s}^{-1}$ &10 \\
$\delta$ & Stribeck exponent & - &2 \\
$a$ & Pressure parameter & - & 0.1 \\
$\varepsilon$ & Regularization parameter & $\textnormal{m}^2\,\textnormal{s}^{-2}$ & $0$ \\
\hline
\end{tabular} }
\label{tab:param1}
\end{table}

\subsection{Mathematical properties}\label{sect:NMEPdistr}

In analogy to what done in Section~\ref{sect:diss} for the lumped variant, the present section investigates the mathematical properties of the distributed FrDB model concerning the mild solution of~\eqref{eq:PDEfric}. Specifically, stability and passivity are analyzed in Section~\ref{sect:stability} and~\ref{sect:diss}, respectively.

\subsubsection{Stability}\label{sect:stability}
Stability is discussed concerning both the input-to-state and input-to-output behaviors. For the system~\eqref{eq:PDEfric}-\eqref{eq:outputDistr}, the notions of \emph{input-to-state} and \emph{input-to-output stability} are enounced in Definitions~\ref{def:ISS} and~\ref{def:IOS}, respectively, for classical solutions of~\eqref{eq:PDEfric}. It should be mentioned that the formalisms adopted here differ from that employed in Section~\ref{sect:stability} for the lumped counterpart of the FrBD model. This is primarily justified by the additional dynamical features possessed by the distributed formulation: essentially, the spatial dimension of the problem permits obtaining stronger stability estimates than those derived for the lumped description.

\begin{definition}[Input-to-state-stability (ISS)]\label{def:ISS}
The PDE~\eqref{eq:PDEfric} is called (uniformly) \emph{input-to-state stable} (ISS) in the spatial $L^2$-norm if, for all inputs $v \in C^1(\mathbb{R}_{\geq 0};\mathbb{R}) \cap L^\infty(\mathbb{R}_{\geq 0};\mathbb{R})$ and ICs $z_0 \in H^1((0,1);\mathbb{R})$ satisfying the BC~\eqref{eq:PDEfricBC}, there exist functions $\beta \in \mathcal{KL}$ and $\gamma \in \mathcal{K}_\infty$ such that
\begin{align}\label{eq:ISSbetaGamma}
\begin{split}
\norm{z(\cdot,t)}_{L^2((0,1);\mathbb{R})} & \leq \beta\Bigl(\norm{z_0(\cdot)}_{L^2((0,1);\mathbb{R})}, t\Bigr) \\
& \quad + \gamma\Bigl( \norm{v(\cdot)}_\infty\Bigr), \quad t \in \mathbb{R}_{\geq 0}.
\end{split}
\end{align}
\end{definition}

\begin{definition}[Input-to-output stability (IOS)]\label{def:IOS}
The PDE~\eqref{eq:PDEfric} with output~\eqref{eq:outputDistr} is called \emph{input-to-output stable} (IOS) if, for all inputs $v \in C^1(\mathbb{R}_{\geq 0};\mathbb{R}) \cap L^\infty(\mathbb{R}_{\geq 0};\mathbb{R})$ and ICs $z_0 \in H^1((0,1);\mathbb{R})$ satisfying the BC~\eqref{eq:PDEfricBC}, there exist functions $\beta \in \mathcal{KL}$ and $\alpha \in \mathcal{K}$ such that 
\begin{align}\label{eq:FBBBBBBSS}
\begin{split}
\abs{F\ped{b}(t)} & \leq \beta\Bigl(\norm{z_0(\cdot)}_{L^2((0,1);\mathbb{R})}, t\Bigr)  + \alpha\Bigl(\norm{v(\cdot)}_\infty\Bigr), \quad t\in \mathbb{R}_{\geq 0}.
\end{split}
\end{align}
\end{definition}
The above definitions, inspired from \cite{Khalil,Prieur}, are contingent on the specific system considered in the paper; a modern, exhaustive introduction to the notions of ISS and \emph{integral input-to-state stability} (iISS) for infinite-dimensional systems may be instead found in \cite{Prieur}. Concerning ISS properties, the main result is asserted below.

\begin{lemma}[Input-to-state stability (ISS)]\label{lemma:boundedness}
The PDE~\eqref{eq:PDEfric} is (uniformly) ISS in the spatial $L^2$-norm. 


\begin{proof}
Using the fact that $\dfrac{\mu}{g} \in C^0(\mathbb{R}\times[0,1];(0,1])$, an ISS estimate may be derived exactly as in Lemma 3.4 in \cite{LuGreDistr}.
\end{proof}
\end{lemma}
Lemma~\ref{lemma:boundedness} is also instrumental in showing the IOS of the model\footnote{In this context, it should be observed that, whilst this paper is mainly concerned with ISS and IOS, iISS estimates \cite{Prieur} may also be derived \cite{LuGreDistr}.}: as for the lumped model, bounded deformations should intuitively generate bounded frictional forces. In fact, combining~\eqref{eq:outputDistr} with~\eqref{eq:muXi} gives, for all $v \in C^0(\mathbb{R}_{\geq 0};\mathbb{R})\cap L^\infty(\mathbb{R}_{\geq 0};\mathbb{R})$ and $z(\cdot,t) \in L^2((0,1);\mathbb{R})$,
\begin{align}\label{eq:boundF}
\begin{split}
\abs{F\ped{b}(t)} & \leq L\norm{p(\cdot)\bar{\sigma}_0(\cdot,\cdot)}_\infty\norm{z(\cdot,t)}_{L^2((0,1);\mathbb{R})} \\
& \quad + L\norm{p(\cdot)\bar{\sigma}_2(\cdot,\cdot)}_\infty\norm{v(\cdot)}_\infty, \quad t \in \mathbb{R}_{\geq 0},
\end{split}
\end{align}
which, in conjunction with Lemma~\ref{lemma:boundedness}, implies the boundedness of $F\ped{b}(t)$ on $\mathbb{R}_{\geq 0}$, and hence IOS. This simple result is formalized in Corollary~\ref{corollary:IOS} below.
\begin{corollary}[Input-to-output stability (IOS)]\label{corollary:IOS}
The PDE~\eqref{eq:PDEfric} with output~\eqref{eq:outputDistr} is IOS.
\end{corollary}
\begin{remark}\label{Remark:oododo}
As explained in detail in \cite{LuGreDistr}, using the partial time derivative in the definition of the virtual friction coefficient~\eqref{eq:muVDistr} would preclude a similar result to that asserted by Corollary~\ref{corollary:IOS}. 
\end{remark}
Remark~\ref{Remark:oododo} concludes the stability analysis of the model~\eqref{eq:PDEfric}-\eqref{eq:outputDistr}. The following Section~\ref{sect:diss} examines its passivity properties.

\subsubsection{Passivity}\label{sect:diss}
Compared to the lumped FrDB model, passivity may be preserved for the distributed formulation by imposing additional constraints on the spatially varying model parameters. For classical solutions of the infinite-dimensional system~\eqref{eq:PDEfric}-\eqref{eq:outputDistr}, the concepts of dissipativity and passivity are enounced according to Definition~\ref{def:idss2} below. A more general introduction to these notions, limited to linear systems, is offered in \cite{Zwart}.
\begin{definition}[Dissipativity and passivity]\label{def:idss2}
The system~\eqref{eq:PDEfric}-\eqref{eq:outputDistr22} is called \emph{dissipative} if, for all inputs $v \in C^1([0,T];\mathbb{R})$ and ICs $z_0 \in H^1((0,1);\mathbb{R})$ satisfying the BC~\eqref{eq:PDEfricBC}, there exist a supply rate $w : L^2((0,1);\mathbb{R})\times\mathbb{R} \mapsto \mathbb{R}$ and a storage function $W : L^2((0,1);\mathbb{R}) \mapsto \mathbb{R}_{\geq 0}$ such that
\begin{align}\label{eq:Fvres}
\begin{split}
\int_0^t w\bigl(f\ped{b}(\cdot,t^\prime),v(t^\prime)\bigr) \dif t^\prime \geq W\bigl(z(\cdot,t)\bigr) &-W\bigl(z_0(\cdot)\bigr), \\
& \qquad t \in [0,T].
\end{split}
\end{align}
It is called \emph{passive} if $w(f\ped{b}(\cdot,t),v(t)) = \langle f\ped{b}(\cdot,t),v(t)\rangle_{L^2((0,1);\mathbb{R})} = F\ped{b}(t)v(t)$.
\end{definition}

\begin{lemma}[Passivity]\label{lemma:DissF2}
Suppose that the \emph{velocity-independent micro-stiffness} coefficient $C^1([0,1];\mathbb{R}_{\geq 0}) \ni \Sigma_0 \triangleq \sigma_0p$ satisfies
\begin{align}\label{eq:diffp}
\dod{\Sigma_0(\xi)}{\xi} \leq -\dfrac{\Sigma_0(\xi)}{V(\xi)}\dod{V(\xi)}{\xi}, \quad \xi \in [0,1].
\end{align}
Then, the distributed friction model~\eqref{eq:PDEfric} is passive with storage function
\begin{align}\label{eq:VdissF2}
W\bigl(z(\cdot,t)\bigr) \triangleq \dfrac{L}{2}\int_0^1 \Sigma_0(\xi)z^2(\xi,t) \dif \xi.
\end{align}

\begin{proof}
Using \eqref{eq:VdissF2} and repeating similar calculations as in the proof of Lemma~\ref{lemma:DissF} provides
\begin{align}\label{eq:firstPpPs}
\begin{split}
 & F\ped{b}(t)v(t)-\dot{W}\bigl(z(\cdot,t)\bigr) = \int_0^1f\ped{b}(\xi,t)v(t) \dif \xi \\
&  \quad= L\int_0^1 \Sigma_0(\xi)\sigma_0(\xi)\dfrac{\abs{v(t)}_\varepsilon}{g\bigl(v(t),\xi\bigr)}z^2(\xi,t)\dif \xi  \\
& \qquad +L\int_0^1p(\xi)\sigma_1\bigl(v(t),\xi\bigr)\dfrac{\mu\bigl(v(t)\bigr)}{g\bigl(v(t),\xi\bigr)}v^2(t) \dif \xi\\
& \qquad + L\int_0^1 \Sigma_0(\xi)V(\xi)z(\xi,t)\dpd{z(\xi,t)}{\xi}\dif \xi, \quad t\in (0,T).
\end{split}
\end{align}
Integrating by parts the last term in~\eqref{eq:firstPpPs} and imposing the BC~\eqref{eq:PDEfricBC} yields
\begin{align}
\begin{split}
&  F\ped{b}(t)v(t)   \geq \dot{W}\bigl(z(\cdot,t)\bigr) \\
&  \quad - \dfrac{L}{2}\int_0^1 \biggl( \dod{\Sigma_0(\xi)}{\xi} +\dfrac{\Sigma_0(\xi)}{V(\xi)}\dod{V(\xi)}{\xi} \biggr)z^2(\xi,t)\dif \xi, \\
& \qquad \qquad \qquad \qquad \qquad \qquad \qquad \qquad \qquad  t \in (0,T).
\end{split}
\end{align}
Hence,~\eqref{eq:diffp} implies that
\begin{align}\label{eq:Fv}
\begin{split}
 F\ped{b}(t)v(t) & \geq \dot{W}\bigl(z(\cdot,t)\bigr), \quad t\in (0,T).
\end{split}
\end{align}
Integrating the above~\eqref{eq:Fv} proves~\eqref{eq:Fvres} with $W(z(\cdot,t))$ defined according to~\eqref{eq:VdissF2}.
\end{proof}
\end{lemma}
The inequality~\eqref{eq:diffp} is identical to that derived for the distributed LuGre friction model in \cite{LuGreDistr}, confirming the importance of correct modeling of the micro-stiffness coefficient in ensuring passivity properties. For instance, for a constant transport velocity and velocity-independent normalized micro-stiffness coefficient $\sigma_0(\xi)$, as considered in practical applications \cite{Tsiotras1,Tsiotras2,Deur1,Deur2}, it may be easily concluded from~\eqref{eq:diffp} that the constant pressure distribution $p(\xi) = p_0$ and the decreasing exponential one $p(\xi) = p_0\exp(-a\xi)$ used in \cite{Tsiotras1} qualify as valid choices, with the latter even rendering the model \emph{strictly dissipative} (the inequalities~\eqref{eq:diffp} and~\eqref{eq:Fvres} are both satisfied in a strict sense). Of course, for a constant contact pressure, the distributed and lumped formulations exhibit identical behaviors. On the other hand, for a parabolic pressure distribution $p(\xi) = p_0\xi(1-\xi)$ \cite{Tsiotras1}, it is possible to show that there exists no $\sigma_0 \in C^1([0,1];\mathbb{R}_{>0})$ verifying~\eqref{eq:diffp}. It is also interesting to note that, albeit seemingly artificial, the condition stated by~\eqref{eq:diffp} is also necessary to correctly reproduce some qualitative steady-state rolling contact behaviors that are not captured by constant pressure profiles, as better explained in \cite{Tsiotras2,Deur1,Deur2}. Once again, it is essential to observe that the result obtained in Lemma~\ref{lemma:DissF} requires adopting the total time derivative of the distributed state in the definition of the virtual friction coefficient~\eqref{eq:mu} \cite{LuGreDistr}. 

\begin{remark}
For simplicity, inequalities of the type~\eqref{eq:ISSbetaGamma},~\eqref{eq:FBBBBBBSS}, and~\eqref{eq:Fvres} havel been proven only for classical solutions ($\varepsilon \in \mathbb{R}_{>0}$ in~\eqref{eq:PDEfric}); the extension to mild solutions with $\varepsilon \in \mathbb{R}_{\geq0}$ in~\eqref{eq:PDEfric} may be worked out using using rather technical arguments, and is not pursued in this paper.
\end{remark}

\subsection{Linearized model}\label{sect:LinearDsitr}
The derivation of the linearized version of the distributed FrBD model proceeds with identical steps as those carried out in Section~\ref{sect:linearLump}, which are repeated here for completeness. Specifically, the linearized PDE reads
\begin{subequations}\label{eq:LuGreDistrDelta}
\begin{align}
\begin{split}
& \dpd{\tilde{z}(\xi,t)}{t} + V(\xi)\dpd{\tilde{z}(\xi,t)}{\xi} = -\dfrac{\sigma_0\bigl(v(t),\xi\bigr)\abs{v^\star}_\varepsilon}{g(v^\star,\xi)}\tilde{z}(\xi,t) \\
& \qquad \qquad + H_1\bigl(z^\star(\xi),v^\star,\xi\bigr)\tilde{v}(t), \quad(\xi,t) \in (0,1)\times(0,T), \label{eq:LugrePDEDelta}
\end{split} \\
& \tilde{z}(0,t)  = 0, \quad t \in (0,T), \label{eq:LugreVCDelta}\\
& \tilde{z}(\xi,0) = \tilde{z}_{0}(\xi), \quad \xi \in (0,1),
\end{align}
\end{subequations}
and the corresponding linearized friction coefficient and infinitesimal bristle force read
\begin{subequations}
\begin{align}\label{eq:mu_delta}
\tilde{\mu}\ped{b}(\xi,t) & \triangleq \bar{\sigma}_0(v^\star,\xi)\tilde{z}(\xi,t) + H_2\bigl(z^\star(\xi),v^\star,\xi\bigr)\tilde{v}(t), \\
\tilde{f}\ped{b}(\xi,t) & \triangleq L\tilde{\mu}\ped{b}(\xi,t)p(\xi),
\end{align}
\end{subequations}
where \begin{small}
\begin{subequations}
\begin{align}
\begin{split}
 H_1\bigl(z(\xi),v,\xi\bigr) & \triangleq \dfrac{\mu(v)}{g(v,\xi)}\biggl(1-\dfrac{v}{g(v,\xi)}\dpd{g(v,\xi)}{v}\biggr)\\
& \quad+\dfrac{v}{g(v,\xi)}\dod{\mu(v)}{v}-\dpd{\sigma_0(v,\xi)}{v}\dfrac{\abs{v}_\varepsilon}{g(v,\xi)}\\
& \quad - \dfrac{\sigma_0(v,\xi)z(\xi)}{g(v,\xi)}\biggr(\dod{\abs{v}_\varepsilon}{v} - \dfrac{\abs{v}_\varepsilon}{g(v,\xi)}\dpd{g(v,\xi)}{v}\biggl), 
\end{split}\\
\begin{split}
H_2\bigl(z(\xi),v,\xi\bigr) & \triangleq \dpd{\bar{\sigma}_0(v,\xi)}{v}z(\xi) + \dpd{\bar{\sigma}_2(v,\xi)}{v}v + \bar{\sigma}_2(v,\xi),
\end{split}
\end{align}
\end{subequations}\end{small}
and\begin{small}
\begin{subequations}
\begin{align}
\begin{split}
 \dpd{\bar{\sigma}_0(v,\xi)}{v} & =\dpd{\sigma_0(v,\xi)}{v}\biggl(1-\dfrac{\sigma_1(v,\xi)\abs{v}_\varepsilon}{g(v,\xi)}\biggr) \\
& \quad -\dfrac{\sigma_0(v,\xi)\abs{v}_\varepsilon}{g(v)}\dpd{\sigma_1(v,\xi)}{v}\\
& \quad  -\dfrac{\sigma_0(v,\xi)\sigma_1(v,\xi)}{g(v,\xi)}\biggl(\dod{\abs{v}_\varepsilon}{v}-\dfrac{\abs{v}_\varepsilon}{g(v,\xi)}\dpd{g(v,\xi)}{v}\biggr), 
\end{split} \\
\begin{split}
\dpd{\bar{\sigma}_2(v,\xi)}{v} & = \dfrac{\sigma_1(v,\xi)}{g^2(v,\xi)}\biggl(g(v,\xi)\dod{\mu(v)}{v}-\mu(v)\dpd{g(v,\xi)}{v}\biggr) \\
& \quad + \dfrac{\mu(v)}{g(v,\xi)}\dpd{\sigma_1(v,\xi)}{v}.
\end{split}
\end{align}
\end{subequations}\end{small}
The corresponding linearized force is formally given by
\begin{align}\label{eq:Flin1}
\tilde{F}\ped{b}(t) = \bigl(\tilde{\mathscr{K}}(\tilde{z},\tilde{v})\bigr)(t) \triangleq \int_0^1\tilde{f}\ped{b}(\xi,t) \dif \xi.
\end{align}
By visual inspection, it is straightforward to conclude that, as for its semilinear counterpart, the linearized model described by~\eqref{eq:LuGreDistrDelta}-\eqref{eq:Flin1} also enjoys ISS and IOS properties.

\section{Experiments and validation}\label{sect:exper}
The behavior of the FrBD model predicted in simulation is discussed in Section~\ref{sect:Numerical}, whereas Section~\ref{sect:valid} addresses its experimental validation. 

\subsection{Numerical experiments}\label{sect:Numerical}
The present section replicates some of the experiments conducted in \cite{Olsson}, and aimed at assessing the dynamical and tribological behaviors of the LuGre-like FrBD model. Specifically, the following phenomena are investigated: pre-sliding behavior, frictional lag hysteresis, and stick-slip dynamics. The simulation results reported below are restricted to the lumped version of the FrBD model.

\subsubsection{Pre-sliding displacement}

Microscopic motion occurs to build up the reactive friction force before break-away. This initial phase is known as \emph{pre-sliding displacement}. Courtney-Pratt and Eisner demonstrated that, under external forces smaller than the breakaway threshold, friction behaves like a nonlinear spring \cite{CPE}. To examine whether the FrBD model replicates this behavior, a first set of simulations was conducted using the model parameters listed in Table~\ref{tab:param2}. Specifically, a unit mass was exposed to an external sinusoidal force with an amplitude amounting to 90\% of the breakaway force, and different excitation frequencies $\omega = 1$, 5, and 10 Hz. The normalized resulting friction force $\mu\ped{b}$ is plotted against the displacement $x$ in Figure~\ref{fig:Fx}. The predicted response aligns qualitatively with the experimental findings of Courtney-Pratt and Eisner \cite{CPE}, and closely resembles that obtained using the LuGre model \cite{Olsson}. Figure~\ref{fig:Fx} was produced using the model parameters listed in Table~\ref{tab:param2}.


\begin{figure}
\centering
\includegraphics[width=1\linewidth]{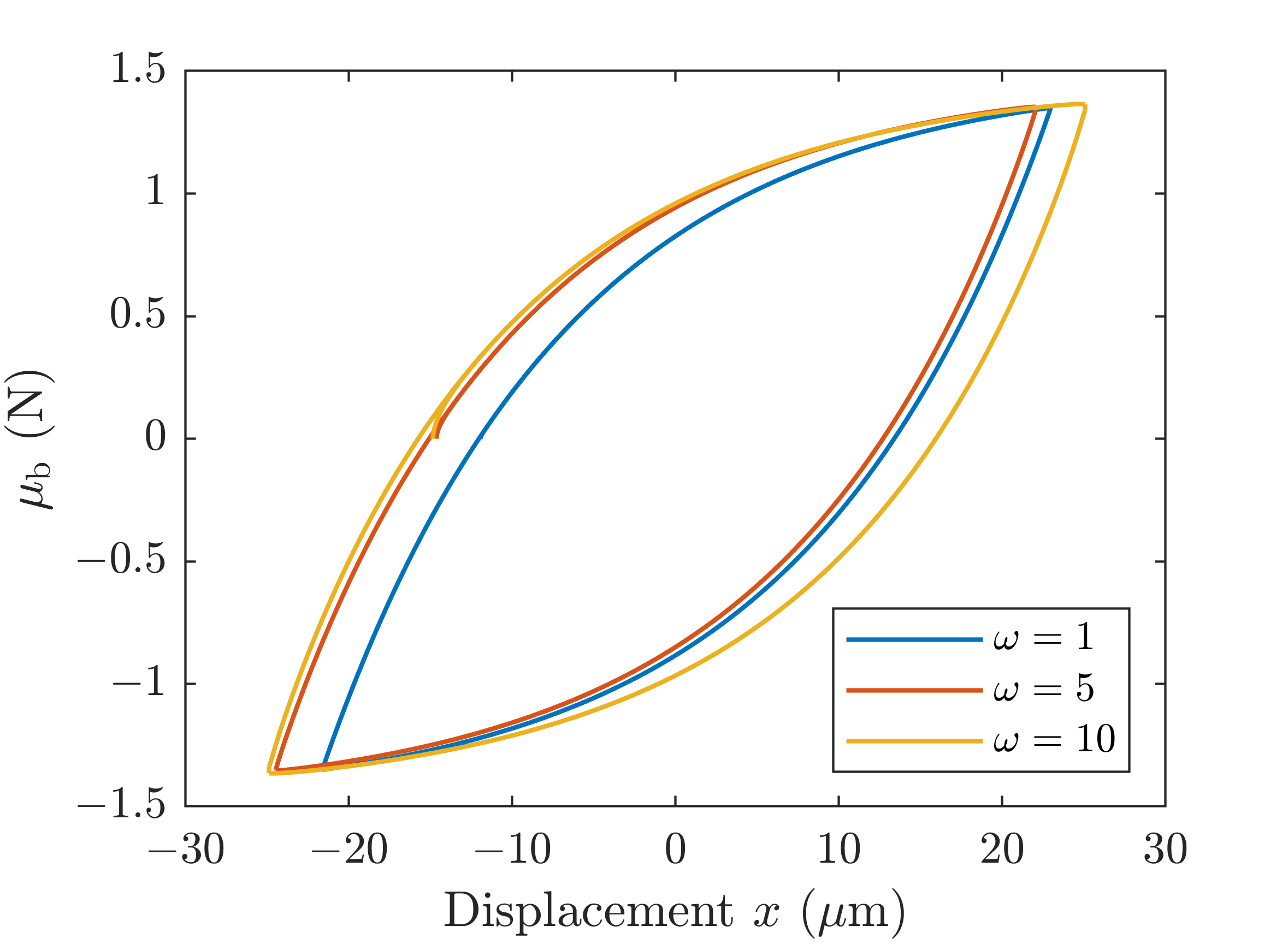} 
\caption{Pre-sliding displacement obtained by simulating the FrBD model, for different values of excitation frequencies $\omega = 1$, 5, and 10 Hz. Model parameters as in Table~\ref{tab:param2}.}
\label{fig:Fx}
\end{figure}

\begin{table}[h!]\centering 
\caption{Model parameters}
{\begin{tabular}{|c|c|c|c|}
\hline
Parameter & Description & Unit & Value \\
\hline 
$\sigma_0$ & Normalized micro-stiffness & $\textnormal{m}^{-1}$ & $10^4$ \\
$\sigma_1$ & Normalized micro-damping & $\textnormal{s}\,\textnormal{m}^{-1}$ &64.5 \\
$\sigma_2$ & Normalized viscous damping & $\textnormal{s}\,\textnormal{m}^{-1}$ &0.04 \\
$\mu\ped{d}$ & Dynamic friction coefficient & - &1 \\
$\mu\ped{s}$ & Static friction coefficient & - &1.5 \\
$v\ped{S}$ & Stribeck velocity & $\textnormal{m}\,\textnormal{s}^{-1}$ &0.01 \\
$\delta$ & Stribeck exponent & - &2 \\
$\varepsilon$ & Regularization parameter & $\textnormal{m}^2\,\textnormal{s}^{-2}$ & $0$ \\
\hline
\end{tabular} }
\label{tab:param2}
\end{table}

\subsubsection{Frictional lag}

The dynamic behavior of friction under varying velocity during unidirectional motion was investigated by Hess and Soom \cite{Hess}. Their experiments revealed the existence of hysteretic phenomena governing the relationship between friction force and velocity: the friction force is lower when the velocity decreases than when it increases. Moreover, the hysteresis loop widens or narrows as the rate of velocity change increases. This behavior cannot be explained by classical friction models but is accurately captured by the LuGre model. Clearly, for $\sigma_1(v) = 0$, and sufficiently small values of the viscous friction term in~\eqref{eq:fViscOOO}, which are typical in practical applications, the lumped LuGre formulation is equivalent to the FrBD model, which indeed exhibits similar behavior to those documented in \cite{Olsson}, and not reported here for brevity.

For nonzero values of the damping coefficient, the hysteresis loop is typically reversed for sufficiently large Stribeck velocities $v\ped{S}$, consistently with the predictions of the LuGre model. The normal hysteretic behavior is, however, recovered for sufficiently low values of $v\ped{S}$. Figure~\ref{fig:Fv} illustrates the hysteresis curves produced by imposing relative velocities with different frequencies $\omega = 25$, 50, and 100 Hz, respectively. The maximum attained normalized frictional force $\mu\ped{b}$ appears to decrease with the frequency, whereas the hysteresis cycles become narrower, coherently with the discussion above.

\begin{figure}
\centering
\includegraphics[width=1\linewidth]{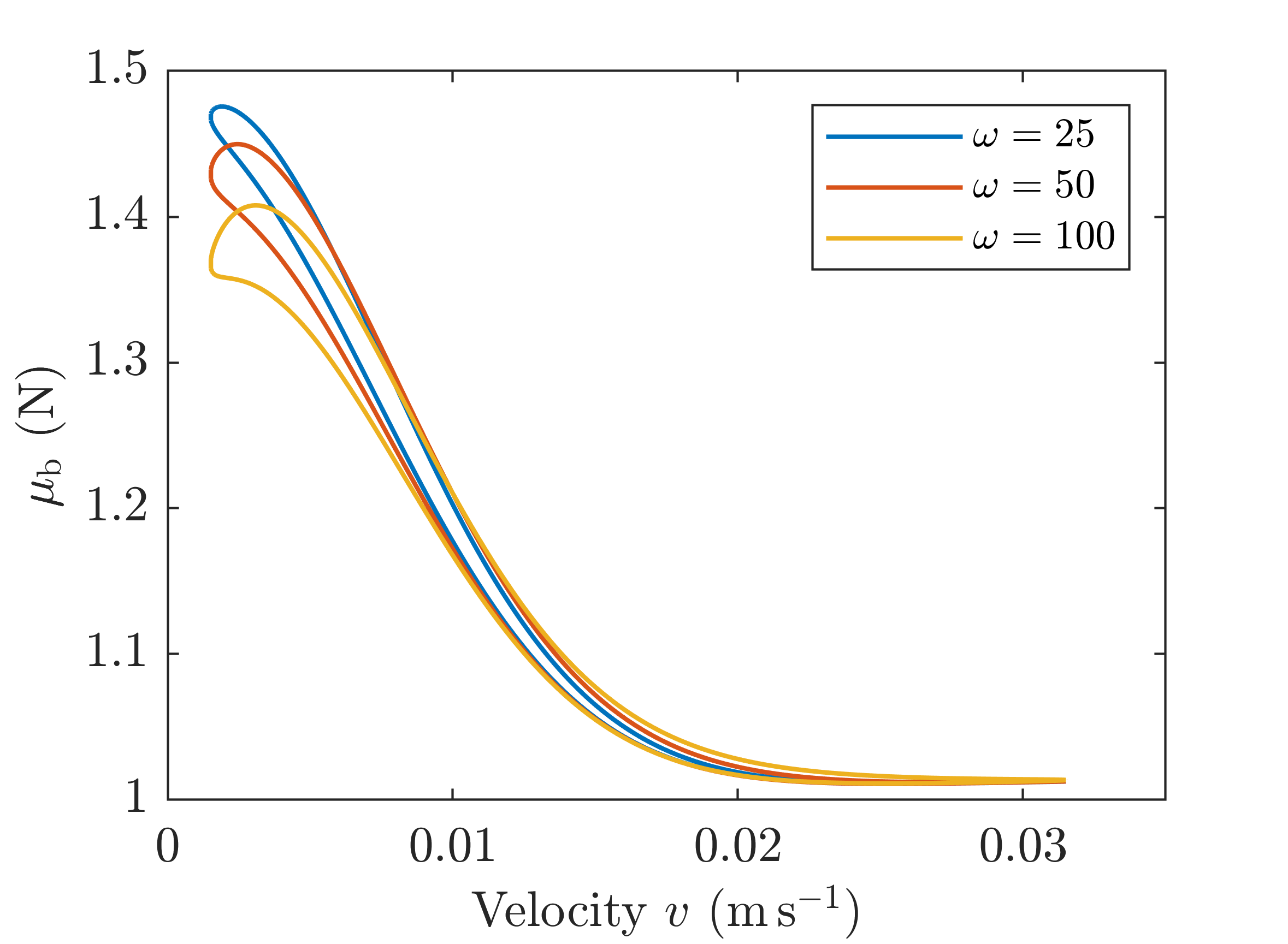} 
\caption{Hysteresis exhibited by the FrBD model for a bristle sliding with sliding velocities at different frequencies $\omega = 25$, 50, and 100 Hz. Model parameters as in Table~\ref{tab:param2}.}
\label{fig:Fv}
\end{figure}

\subsubsection{Stick-slip motion}
The last numerical experiment considered in this paper is concerned with stick-slip motion, a well-known phenomenon that typically arises at low velocities and is characterized by an irregular alternation between sticking and slipping phases. This jerky behavior is particularly undesirable in precision applications, such as machine tools, where it not only degrades control performance but also induces unwanted vibrations and noise. Understanding the mechanisms behind stick-slip motion is thus essential for its mitigation. The phenomenon originates from the friction force being higher at zero velocity than at small nonzero velocities. As motion initiates, the friction force drops abruptly, resulting in a sudden acceleration and contributing to the onset of slip. A standard setup employed to study numerically stick-slip motion is schematized in Figure~\ref{fig:massSpring}, and consists of a mass attached to a linear spring whose free end moves with imposed velocity $v\ped{ref}(t) \in \mathbb{R}_{>0}$. 
\begin{figure}
\centering
\includegraphics[width=0.9\linewidth]{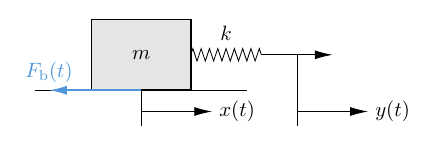} 
\caption{A schematic of the simulation setup described by~\eqref{eq:systemSpring}.}
\label{fig:massSpring}
\end{figure}
For the system illustrated in Figure~\ref{fig:massSpring}, the governing equations of motion may be inferred to be
\begin{subequations}\label{eq:systemSpring}
\begin{align}
m\ddot{x}(t) & = k\bigl(y(t)-x(t)\bigr)-F\ped{b}(t), \\
\dot{y}(t) & = v\ped{ref}(t), && t \in (0,T),
\end{align}
\end{subequations}
where $x(t) \in \mathbb{R}$ denotes the displacement of the mass, $y(t) \in \mathbb{R}_{\geq 0}$ that of the unconstrained end of the spring, and the bristle force $F\ped{b}(t)$ is given by an appropriate friction model.
Figure~\ref{fig:sim1_0} illustrates the numerical results obtained by simulating~\eqref{eq:systemSpring} using both the lumped FrBD and LuGre friction models, with $v\ped{ref} = 0.1$ $\textnormal{m}\,\textnormal{s}^{-1}$, $m = 1$ kg, and $k = 2$ N, and model parameters as in Table~\ref{tab:param2}. For both descriptions, the viscous term is modeled as $\mu\ped{v}(v) = \sigma_2\abs{v}$. Qualitatively, both the FrBD and LuGre models produce a similar phenomenology: initially, the mass remains at rest whilst the spring force increases linearly with the displacement $y(t)$. During this sticking phase, the friction force balances the spring force, preventing the motion of the block. Once the applied force reaches the static friction threshold, the mass begins to slide, and the friction force abruptly drops below the dynamic (Coulomb) level. Consequently, the mass accelerates, but as it moves, the spring contracts, reducing the spring force: the mass decelerates and eventually comes to a stop. This cycle of stick-slip motion is repeated several times. It is particularly interesting to observe that, whilst the displacement $x(t)$ and relative velocity $v(t)$ exhibit similar trends for both the FrBD and LuGre models, the former predicts a bristle force $ F\ped{b}(t)$ that closely follows the evolution of the displacement $z(t)$. This behavior can be attributed to the stronger influence of the normalized micro-damping coefficient on the dynamics of $z(t)$.
\begin{figure}
\centering
\includegraphics[width=1\linewidth]{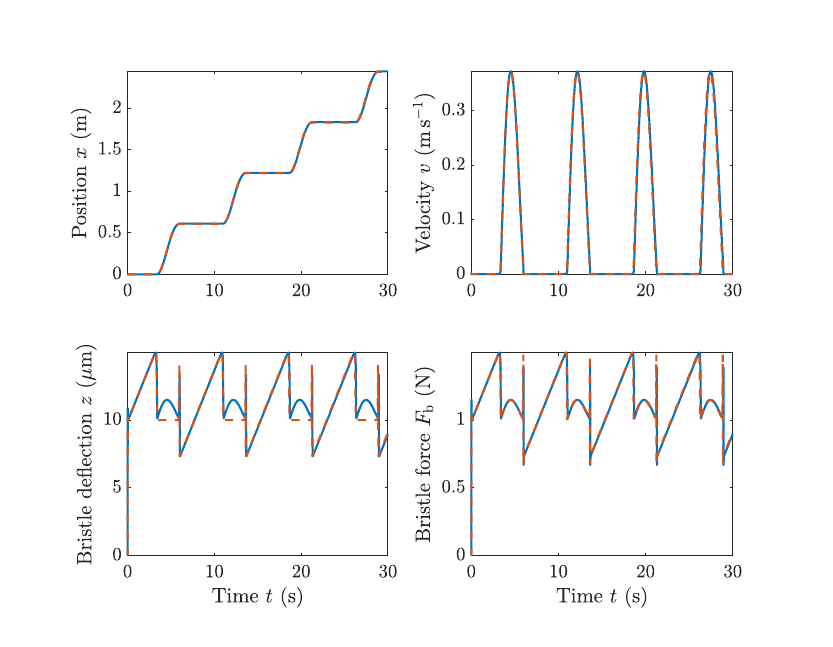} 
\caption{Simulation results for the mechanical system described by~\eqref{eq:systemSpring}, obtained using the FrBD (blue line) and LuGre (dashed orange line) friction models. Model parameters as in Table~\ref{tab:param2}.}
\label{fig:sim1_0}
\end{figure}


\subsection{Experimental validation}\label{sect:valid}
To experimentally validate the dynamic behavior of the lumped FrBD model, a more realistic application is considered involving a diaphragm valve system. This system couples electro-pneumatic, mechanical, and hydraulic subsystems \cite{Fellipe1,Fellipe2}. Variations in the chamber pressure act on the diaphragm, which converts them into a force that accelerates the mechanical system until a new equilibrium is reached. As the stem position changes, the resulting variation in pressure drop across the valve alters the flow rate of the hydraulic system. A mathematical model of the system is given by \cite{Fellipe1,Fellipe2}
\begin{subequations}\label{eq:systemSpringD}
\begin{align}
m\ddot{x}(t) & = S\ped{a}P(t)-kx(t)-F\ped{b}(t)-F_0, \label{eq:stem}\\
\dot{P}(t) & = \dfrac{K_P\mathrm{OP}(t) + P\ped{min}-P(t)}{\tau}, \label{eq:IP} && t \in (0,T),
\end{align}
\end{subequations}
where $x(t) \in \mathbb{R}$ denotes the stem position, $P(t) \in \mathbb{R}$ the pressure in the valve diaphragm chamber, and $\mathrm{OP}(t) \in [0,100]$ is the input expressed in percentage. Equations~\eqref{eq:stem} and~\eqref{eq:IP} describe the dynamics of the stem mass and I/P converter, respectively. In~\eqref{eq:systemSpringD}, $m\in \mathbb{R}_{>0}$ is the mass of the moving part, $S\ped{a}\in \mathbb{R}_{>0}$ represents the cross-sectional area of the valve, $k \in \mathbb{R}_{>0}$ the elastic spring constant, $F_0 \in \mathbb{R}$ indicates the applied preload, $K_P\in \mathbb{R}_{>0}$ is the pressure gain, $P\ped{min} \in \mathbb{R}_{\geq 0}$ denotes a minimum pressure level, and $\tau \in \mathbb{R}_{>0}$ is the time constant of the I/P converter.

Starting from~\eqref{eq:systemSpringD}, the FrBD model was validated using the open-access experimental dataset from \cite{Fellipe3}, which provides chamber pressure $P(t)$ and stem displacement $x(t)$ measurements. Two tests were considered: \textbf{Test 1} with ramp input $\mathrm{OP}(t)$, and \textbf{Test 2} with sinusoidal input. Except for $\sigma_1$, which was recalibrated using a genetic algorithm, the model parameters were retained from \cite{Fellipe1}. The values for both tests are reported in Table~\ref{tab:paramExp}.

Figure~\ref{fig:sim1} compares the simulated responses of~\eqref{eq:systemSpringD} with the FrBD and LuGre friction models against the experimental measurements, along with the trends of the internal variables $v(t)$ and $z(t)$. Both models show good agreement with the data, with the FrBD achieving slightly better performance indicators (Table~\ref{tab:perf}) in terms of prediction error for stem displacement $x(t)$. Apart from minor numerical discrepancies in the third and last panels of Figure~\ref{fig:sim1}, which mainly relate to sign reversals for the velocity $v(t)$, the FrBD reproduces system dynamics closely, whilst maintaining similar behavior to the LuGre model under the chosen parameter sets.
\begin{figure*}
\centering
\includegraphics[width=1\linewidth]{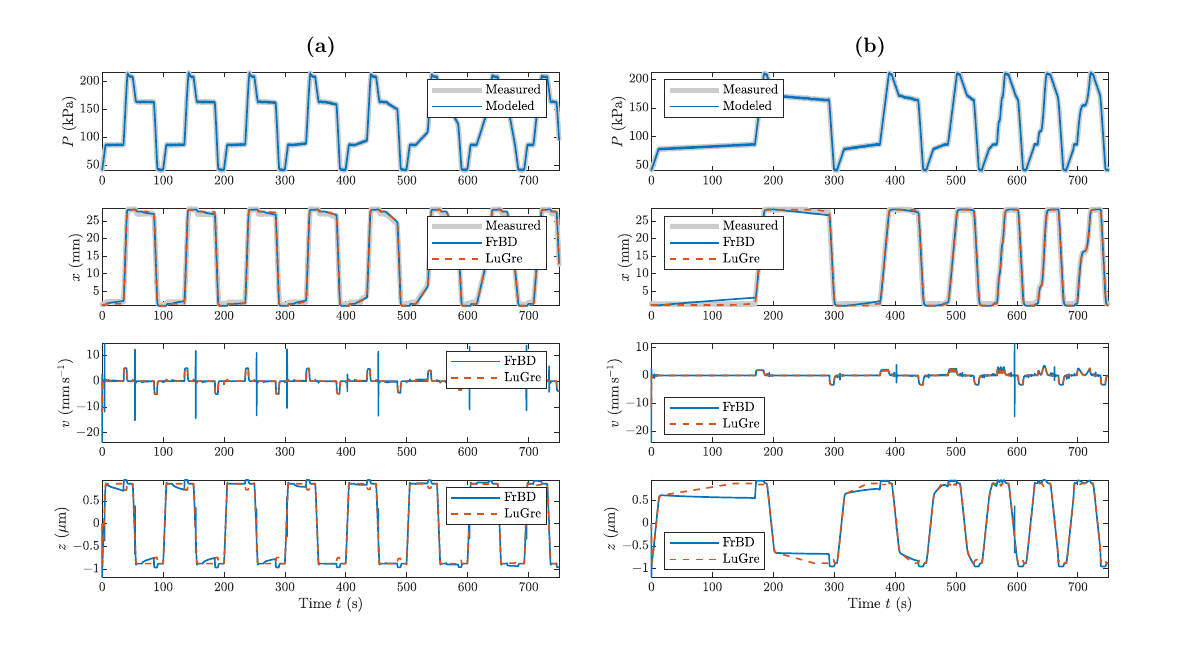} 
\caption{Results for the mechanical system described by~\eqref{eq:systemSpringD}, obtained using the FrBD (blue line) and LuGre (dashed orange line) friction models: \textbf{(a) Test 1}; \textbf{(b) Test 2}. Model parameters as in Table~\ref{tab:paramExp}.}
\label{fig:sim1}
\end{figure*}


\begin{table}[]
\centering
\caption{Model parameter}
\begin{tabular}{@{}|llll|@{}}
\toprule
\multicolumn{4}{|c|}{\textit{Valve parameters}}                                                                                                                                                                        \\
\midrule
\multicolumn{1}{|l|}{\textbf{Parameter}}          & \multicolumn{1}{l|}{\textbf{Description}}                                   & \multicolumn{1}{l|}{\textbf{Unit}}                             & \textbf{Value}      \\
\midrule
\multicolumn{1}{|l|}{$m$}                         & \multicolumn{1}{l|}{Mass}                                                   & \multicolumn{1}{l|}{kg}                                        & 1.6                 \\
\multicolumn{1}{|l|}{$P\ped{min}$}                & \multicolumn{1}{l|}{Minimum pressure}                                       & \multicolumn{1}{l|}{Pa}                                        & 41276.40            \\
\multicolumn{1}{|l|}{$K_P$}                       & \multicolumn{1}{l|}{Pressure gain}                                          & \multicolumn{1}{l|}{Pa}                                        & 1666.49             \\
\multicolumn{1}{|l|}{\multirow{2}{*}{$\tau$}}      & \multicolumn{1}{l|}{Time constant (\textbf{Test 1})}                                 & \multicolumn{1}{l|}{\multirow{2}{*}{s}}                        & 0.933               \\
\multicolumn{1}{|l|}{}                            & \multicolumn{1}{l|}{Time constant (\textbf{Test 2})}                                 & \multicolumn{1}{l|}{}                                          & 0.425               \\
\multicolumn{1}{|l|}{$S\ped{a}$}                       & \multicolumn{1}{l|}{Cross-sectional area}                                   & \multicolumn{1}{l|}{$\textnormal{m}^2$}                        & $445\cdot 10^{-4}$  \\
\multicolumn{1}{|l|}{$k$}                         & \multicolumn{1}{l|}{Spring constant}                                        & \multicolumn{1}{l|}{$\textnormal{N}\,\textnormal{m}^{-1}$}     & 203495.8            \\
\multicolumn{1}{|l|}{$F_0$}                       & \multicolumn{1}{l|}{Preload}                                                & \multicolumn{1}{l|}{N}                                         & 2578.3              \\
\midrule
\multicolumn{4}{|c|}{\textit{Friction parameters}}                                                                                                                                                                     \\
\midrule 
\multicolumn{1}{|l|}{\textbf{Parameter}}          & \multicolumn{1}{l|}{\textbf{Description}}                                   & \multicolumn{1}{l|}{\textbf{Unit}}                             & \textbf{Value}      \\
\midrule
\multicolumn{1}{|l|}{$\sigma_0$}                  & \multicolumn{1}{l|}{Normalized micro-stiffness}                     & \multicolumn{1}{l|}{$\textnormal{m}^{-1}$}                     & $6.82\cdot 10^{7}$  \\
\multicolumn{1}{|l|}{$\sigma_1$}                             & \multicolumn{1}{l|}{Normalized micro-damping}                       & \multicolumn{1}{l|}{}                                          & 701.97              \\
\multicolumn{1}{|l|}{$\sigma_2$}                             & \multicolumn{1}{l|}{Normalized viscous damping}                     & \multicolumn{1}{l|}{}                                          & $2.97\cdot 10^{3}$  \\
\multicolumn{1}{|l|}{$\mu\ped{d}$}                & \multicolumn{1}{l|}{Dynamic friction coefficient}                           & \multicolumn{1}{l|}{-}                                         & 39.73               \\
\multicolumn{1}{|l|}{$\mu\ped{s}$}                & \multicolumn{1}{l|}{Static friction coefficient}                            & \multicolumn{1}{l|}{-}                                         & 59.86               \\
\multicolumn{1}{|l|}{$v\ped{S}$}                  & \multicolumn{1}{l|}{Stribeck velocity}                                      & \multicolumn{1}{l|}{$\textnormal{m}\,\textnormal{s}^{-1}$}     & $6.42\cdot 10^{-3}$ \\
\multicolumn{1}{|l|}{$\delta$}                  & \multicolumn{1}{l|}{Stribeck exponent}                                      & \multicolumn{1}{l|}{-}     & 2 \\
\multicolumn{1}{|l|}{$\varepsilon$}                  & \multicolumn{1}{l|}{Regularization parameter}                                      & \multicolumn{1}{l|}{$\textnormal{m}^2\,\textnormal{s}^{-2}$}     & 0 \\
\bottomrule            
\end{tabular}
\label{tab:paramExp}
\end{table}

\begin{table}[]
\caption{Performance indicators}
\begin{tabular}{|l|ll|ll|}
\hline
\multicolumn{1}{|c|}{\textbf{}}          & \multicolumn{2}{c|}{\textbf{FrBD}}                                              & \multicolumn{2}{c|}{\textbf{LuGre}}                                              \\ \hline
\multicolumn{1}{|c|}{\textbf{Indicator}} & \multicolumn{1}{c|}{\textbf{Test 1}}     & \multicolumn{1}{c|}{\textbf{Test 2}} & \multicolumn{1}{c|}{\textbf{Test 1}}      & \multicolumn{1}{c|}{\textbf{Test 2}} \\ \hline
\textbf{RMSE}                            & \multicolumn{1}{l|}{$4.45\cdot 10^{-4}$} & $7.06\cdot 10^{-4}$                  & \multicolumn{1}{l|}{$4.77 \cdot 10^{-4}$} & $4.69 \cdot 10^{-4}$                 \\ \hline
\textbf{Mean}                            & \multicolumn{1}{l|}{$2.04\cdot 10^{-5}$} & $-6.72\cdot 10^{-5}$                  & \multicolumn{1}{l|}{$2.94\cdot 10^{-5}$}  & $1.27\cdot 10^{-4}$                  \\ \hline
\textbf{Std}                             & \multicolumn{1}{l|}{$4.45\cdot 10^{-4}$} & $7.03\cdot 10^{-4}$                  & \multicolumn{1}{l|}{$4.76 \cdot 10^{-4}$} & $4.51 \cdot 10^{-4}$                 \\ \hline
\textbf{Max}                             & \multicolumn{1}{l|}{0.0026}              & 0.0021                               & \multicolumn{1}{l|}{0.0026}               & 0.0021                               \\ \hline
\end{tabular}
\label{tab:perf}
\end{table}

\section{Conclusions}\label{ref:Concl}
This paper introduced a novel class of first-order, rate-dependent friction (FrBD) models, derived physically through an inversion of the friction curve. The proposed approach offers valuable guidelines for developing rate-dependent friction models grounded on solid physical principles, as opposed to relying on empirical arguments. The general methodology outlined in the article was explicitly applied to derive a lumped formulation, inspired by the structure of the LuGre model, which was thoroughly analyzed in terms of its stability and passivity properties. A distributed parameter extension was also developed, which is more suitable for the modeling of rolling contact phenomena. Similar to its lumped counterpart, the mathematical properties of the distributed formulation were rigorously examined, revealing both striking analogies and key differences with the LuGre model. In particular, the proposed FrBD formulation appears to naturally possess passivity properties, whereas the LuGre requires the adoption of a velocity-dependent damping coefficient. Differences and analogies were also evident in the numerical validation of the model, which aimed to assess its tribological behavior in relation to typical phenomena observed in mechanical systems, such as pre-sliding displacement, frictional lag, and stick-slip. The lumped FrBD model was also validated experimentally, considering a realistic application involving a diaphragm valve system. 

Future research may focus on deriving equivalent FrBD models by postulating alternative expressions for the rheological behavior of the bristle element and the friction coefficient as a function of the sliding velocity. Additionally, experimental validation in the context of more complex mechanical and mechatronic systems is necessary, especially concerning the distributed FrBD variant. Finally, the proposed model – or family of models – holds potential for applications in control and observer design.

\section*{Acknowledgment}
The authors gratefully acknowledge financial support from the project FASTEST (Reg. no. 2023-06511), funded by the Swedish Research Council.

\section*{Declaration of interest}
Declaration of interest: none.


\end{document}